\documentclass{article}

\usepackage{arxiv}

\usepackage[utf8]{inputenc} 
\usepackage[T1]{fontenc}    
\usepackage{hyperref}       
\usepackage{url}            
\usepackage{booktabs}       
\usepackage{amsfonts}       
\usepackage{nicefrac}       
\usepackage{microtype}      
\usepackage{graphicx}
\usepackage{algorithmic}
\usepackage{subfigure}
\usepackage{booktabs} 

\usepackage{graphicx}
\usepackage{tikz}
\usepackage{amsthm}
\usepackage{thmtools}
\usepackage{thm-restate}
\usepackage{hyperref}
\usepackage{amsmath}
\usepackage{cleveref}
\usepackage{natbib}
\usepackage{amsfonts}

\usetikzlibrary{arrows,fit}
\usepackage[]{algorithm2e}
\usepackage{mathtools}

\newtheorem{theorem}{Theorem}
\newtheorem{definition}{Definition}
\newtheorem{lemma}{Lemma}

\newtheorem{conjecture}{Conjecture}
\DeclareMathOperator{\leaves}{leaves}
\DeclareMathOperator{\nonleaves}{non-leaves}
\DeclareMathOperator{\rev}{rev}
\DeclareMathOperator{\val}{val}
\DeclareMathOperator{\cost}{cost}
\DeclareMathOperator{\mergecost}{merge-cost}
\DeclareMathOperator{\mergerev}{merge-rev}
\DeclareMathOperator{\clustercost}{clustering-cost}
\DeclareMathOperator{\clusterrev}{clustering-rev}

\newtheorem{innercustomthm}{Theorem}
\newenvironment{customthm}[1]
  {\renewcommand\theinnercustomthm{#1}\innercustomthm}
  {\endinnercustomthm}

\newenvironment{customlem}[1]
  {\renewcommand\theinnercustomthm{#1}\innercustomthm}
  {\endinnercustomthm}

\tikzset{
  treenode/.style = {align=center, inner sep=0pt, text centered,
    font=\sffamily},
  arn_bwb/.style = {treenode, circle, black, font=\sffamily\bfseries, draw=black,
    fill=white, text width=2em},
  arn_bwbb/.style = {treenode, circle, black, font=\sffamily\bfseries, draw=black,
    fill=white, text width=2em, very thick},
  arn_wbb/.style = {treenode, circle, white, font=\sffamily\bfseries, draw=black,
    fill=black, text width=2em},   
  arn_wbr/.style = {treenode, circle, white, font=\sffamily\bfseries, draw=red,
    fill=black, text width=2em, very thick}, 
  arn_bl/.style = {treenode, circle, black, font=\sffamily\bfseries, draw=blue,
    fill=white, text width=2em, very thick}, 
  arn_b/.style = {treenode, circle, black, draw=blue, 
    text width=2em, very thick},
  arn_x/.style = {treenode, rectangle, draw=black,
    minimum width=0.5em, minimum height=0.5em}
}
\newcommand{\matchcluster}{In the revenue context (where edge weights are data similarity), with $\widetilde{O}(n)$ machine space, Matching Affinity Clustering achieves (whp):
\begin{itemize}
\item a $(1/3-\epsilon)$-approximation for revenue in $O(\log (n)\log\log (n) \cdot (1/\epsilon)^{O(1/\epsilon)})$ rounds when $n=2^N$,
\item and a $(1/9-\epsilon)$-approximation for revenue in $O(\log(nW)\log\log(n)\cdot(1/\epsilon)^{O(1/\epsilon)})$ rounds in general.
\end{itemize}
}

\newcommand{\matchclusterval}{Assume there exists an MPC algorithm that achieves an $\alpha$-approximation for minimum weight $k$-sized matching whp in $O(f(n))$ rounds and $\widetilde{O}(n)$ machine space. In the value context (where edge weights are data distances) and in $O(f(n)\log(n))$ rounds with $\widetilde{O}(n)$ machine space, Matching Affinity Clustering achieves (whp):
\begin{itemize}
\item a $\frac23\alpha$-approximation for value when $n=2^N$,
\item and a $\frac13\alpha$-approximation for value in general.
\end{itemize}
}

\newcommand{\affinity}{Affinity Clustering cannot achieve better than a $O(1/n)$-factor approximation for revenue or value.}

\newcommand{\affinityrev}{There exists a family of graphs on which Affinity Clustering cannot achieve better than a $O(1/n)$-factor approximation for revenue.}

\newcommand{\affinityval}{There exists a family of graphs on which Affinity Clustering cannot achieve better than a $O(1/n)$-factor approximation for value.}

\newcommand{\kmatch}{There exists an MPC algorithm for $k$-sized maximum matching with nonnegative edge weights and max edge weight $W$ for $k>n/2$ that achieves a $(1-\epsilon)$-approximation whp in $O(\log(nW)\log\log(n)\cdot (1/\epsilon)^{1/\epsilon})$ rounds and $O(n/polylog(n))$ machine space.}

\newcommand{\half}{After the first round of merges, Matching Affinity Clustering maintains cluster balance (ie, the minimum ratio between cluster sizes) of $1/2$ whp.}

\newcommand{\transform}{Given $\mathcal{C}(G,C^i)$ and $C^{i+1}$ where clusters are all composed of two subclusters in $C^i$, $\mathcal{C}(G, C^{i+1})$ can be computed in the MPC model with $\widetilde{O}(n)$ machine space and one round.}

\newcommand{\lemrev}{Let clusters $C^{i}$ and $C^{i+1}$ be the $i$th and $i+1$th level clusterings found by Matching Affinity Clustering, where $C^0=V$. Let $p$ be the indicator that is 1 if $n$ is not a power of 2. Then the clustering revenue of Matching Affinity Clustering at the $i$th level is at least (whp):
    \begin{align*}
    \clusterrev_G(C^i, C^{i+1}) \geq  2^{3i-2p+1} \left(2^{n-i-1} - 1\right)\sum_{(A,B)\in M_i} w_{\mathcal{C}(G,C^i)}(v_A,v_B).
    \end{align*}}

\newcommand{\costb}{Let $C^{i}$ and $C^{i+1}$ be the $i$th and $i+1$th level clusterings found by Matching Affinity Clustering, where the $i$th step uses matching $M_i\geq (1-\epsilon)M^*$ for maximum matching $M^*$ and $C^0=V$. Then the clustering cost of Matching Affinity Clustering at the $i$th level is at most (whp):
    \begin{align*}\clustercost_G(C^i, C^{i+1}) \leq \frac{2^{2p+1}}{1-\epsilon} \clusterrev_G(C^i, C^{i+1}).
    \end{align*}}
    
\newcommand{\lemapx}{Matching Affinity Clustering obtains a $(1/3-\epsilon)$-approximation for revenue on graphs of size $2^N$, and a $(1/9-\epsilon)$-approximation on general graphs whp.}

\newcommand{\bounds}{Matching Affinity Clustering uses $\widetilde{O}(n)$ space per machine and runs in $ O(\log(n)\allowbreak\log\log (n) \cdot (1/\epsilon)^{O(1/\epsilon)})$ rounds on graphs of size $2^N$, and $O(\log (nW)\log\log (n) \cdot (1/\epsilon)^{O(1/\epsilon)})$ rounds in general.}

\newcommand{\lb}{There is a graph $G$ on which Matching Affinity Clustering achieves no better than a $(1/3 + o(1))$-approximation of the optimal revenue.}

\newcommand{\lbval}{There is a graph $G$ on which Matching Affinity Clustering achieves no better than a $(2/3 + o(1))$-approximation of the optimal revenue.}

\newcommand{\vallem}{Let $T$ be the tree returned by Matching Affinity Clustering in the distance context. Consider any clustering $C$ at some iteration of Matching Affinity Clustering above the first level. Let $C_i$ be the $i$th cluster which merged clusters $A_i$ and $B_i$ in the previous iteration. Say there are $k$ clusters in $C$. Then whp, given an $\alpha$-approximation MPC algorithm for minimum weight $k$-sized matching whp:
\begin{align*}
\frac{\sum_{i=1}^k w(A_i, B_i)}{\sum_{i=1}^k |A_i|\cdot |B_i|} \geq 2\alpha\frac{\sum_{i=1}^k (w(A_i) + w(B_i)) }{\sum_{i=1}^k(|A_i|(|A_i|-1) + |B_i|(|B_i|-1))}
\end{align*}}

\title{Improved Hierarchical Clustering on Massive Datasets with Broad Guarantees}

\author{
  MohammadTaghi Hajiaghayi\\
  Department of Computer Science\\
  University of Maryland, College Park\\
  \texttt{hajiaghayi@gmail.com}
    \And
  Marina Knittel\\
  Department of Computer Science\\
  University of Maryland, College Park\\
  \texttt{mknittel@umd.edu} \\
}

\begin{document}

\maketitle

\begin{abstract}  
Hierarchical clustering is a stronger extension of one of today's most influential unsupervised learning methods: clustering. The goal of this method is to create a hierarchy of clusters, thus constructing cluster evolutionary history and simultaneously finding clusterings at all resolutions. We propose four traits of interest for hierarchical clustering algorithms: (1) empirical performance, (2) theoretical guarantees, (3) cluster balance, and (4) scalability. While a number of algorithms are designed to achieve one to two of these traits at a time, there exist none that achieve all four.

Inspired by Bateni et al.'s scalable and empirically successful Affinity Clustering [NeurIPs 2017], we introduce Affinity Clustering's successor, \textit{Matching Affinity Clustering}. Like its predecessor, Matching Affinity Clustering maintains strong empirical performance and uses \textit{Massively Parallel Communication} as its distributed model. Designed to maintain provably balanced clusters, we show that our algorithm achieves good, constant factor approximations for Moseley and Wang's \textit{revenue} and Cohen-Addad et al.'s \textit{value}. We show Affinity Clustering cannot approximate either function. Along the way, we also introduce an efficient $k$-sized maximum matching algorithm in the MPC model.
\end{abstract}

\section{Introduction}\label{sec:intro}

Clustering is one of the most prominent methods to provide structure, in this case clusters, to unlabeled data. It requires a single parameter $k$ for the number of clusters. Hierarchical clustering elaborates on this structure by adding a hierarchy of clusters contained within superclusters. This problem is unparameterized, and takes in data as a graph whose edge weights represent the similarity or dissimilarity between data points. A hierarchical clustering algorithm outputs a tree $T$, whose leaves represent the input data, internal nodes represent the merging of data and clusters into clusters and superclusters, and root represents the cluster of all data.

Obviously it is more computationally intensive to find $T$ as opposed to a flat clustering. However, having access to such a structure provides two main advantages: (1) it allows a user to observe the data at different levels of granularity, effectively querying the structure for clusterings of size $k$ without recomputation, and (2) it constructs a history of data relationships that can yield additional perspectives. The latter is most readily applied to phylogenetics, where dendrograms depict the evolutionary history of genes and species \citep{phylo}. Hierarchical clustering in general has been used in a number of other unsupervised applications. 
In this paper, we explore four important qualities of a strong and efficient hierarchical clustering algorithm:

\begin{enumerate}
\item \textbf{Theoretical guarantees}. Previously, analysis of hierarchical clustering algorithms has relied on experimental evaluation. While this is one indicator for success, it cannot ensure performance across a large range of datasets. Researchers combat this by considering optimization functions to evaluate broader guarantees \citep{charikar2004,linetal}. One function that has received significant attention recently \citep{charikar2018,cohen-addad} is a hierarchical clustering cost function proposed by \citet{dasgupta}. This function is simple and intuitive, however, \citet{dasguptabounds} showed that it is likely not constant-factor approximable. To overcome this, we examine its dual, revenue, proposed by \citet{dual}, which considers a graph with \textit{similarity}-based edge weights. For \textit{dissimilarity}-based edge weights, we look to \citet{cohen-addad}'s value, another cost-inspired function. We are interested in constant factor approximations for these functions.

\item \textbf{Empirical performance}. As theoretical guarantees are often only intuitive proxies for broader evaluation, it is still important to evaluate the empirical performance of algorithms on specific, real datasets. Currently, \citet{affinity}'s Affinity Clustering remains the state-of-the-art for scalable hierarchical clustering algorithms with strong empirical results. With Affinity Clustering as an inspirational baseline for our algorithm, we strive to preserve and, hopefully, extend Affinity Clustering's empirical success.

\item \textbf{Balance}. One downside of algorithms like Affinity Clustering is that they are prone to creating extremely unbalanced clusters. There are a number of natural clustering problems where balanced clusters are preferable or more accurate for the problem, for example, clustering a population into genders. Some more specific applications include image collection clustering, where balanced clusters can make the database more easily navigable \citep{balanced1}, and wireless sensor networks, where balancing clusters of sensor nodes ensures no cluster head gets overloaded \citep{balanced2}. Here, we define balance as the minimum ratio between cluster sizes. 

\item \textbf{Scalability}. Most current approximations for revenue are serial and do not ensure performance at scale. We achieve scalability through distributed computation. Clustering itself, as well as many other big data problems, has been a topic of interest in the distributed community in recent years \citep{chitnisetal2,chitnisetal,networkcluster}. In particular, hierarchical clustering has been studied by \citet{jinetal2,jinetal}, but only \citet{affinity} has attempted to ensure theoretical guarantees through the introduction of a Steiner-based cost metric. However, they provide little motivation for its use. Therefore, we are interested in evaluating distributed algorithms with respect to more well-founded optimization functions like revenue and value.

For our distributed model, we look to \textit{Massively Parallel Communication} (MPC), which was used to design Affinity Clustering. MPC is a restrictive, theoretical abstraction of \textit{MapReduce}: a popular programming framework famous for its ease of use, fault tolerance, and scalability \citep{mapreduce}. In the MPC model, individual machines carry only a fraction of the data and execute individual computations in rounds. At the end of each round, machines send limited messages to each other. Complexities of interest are the number of rounds and the individual machine space. This framework has been used in the analysis for many large-scale problems in big data, including clustering \citep{im,ludwig,networkcluster}. It is a natural selection for this work.
\end{enumerate}

\subsection{Related Work}
There exist algorithms that can achieve up to two of these qualities at a time. Affinity Clustering, notably, exhibits good empirical performance and scalability using MPC. While \citet{affinity} describe some minor theoretical guarantees for Affinity Clustering, we believe that proving an algorithm's ability to optimize for revenue and value is a stronger and more well-founded result due to their popularity and relation to Dasgupta's cost function. A simple random divisive algorithm proposed by \citet{dasguptabounds} was shown to achieve a $1/3$ expected approximation for revenue and can be efficiently implemented using MPC. However, it is notably nondeterministic, and we show that it does not exhibit good empirical performance. Similarly, balanced partitioning may achieve balanced clusters, but it is unclear whether it is scalable, and it has not been shown to achieve strong theoretical guarantees.

For both revenue and value, Average Linkage achieves near-state-of-the-art $1/3$ and $2/3$-approximations respectively~\citep{dual, cohen-addad}. \citet{linkage} marginally improves these factors to $1/3+\epsilon$ for revenue and $2/3+\epsilon$ for value, through semi-definite programming (SDP, a non-distributable method). However, since value and revenue both strove to characterize Average Linkage's optimization goal, and this was only marginally beat by an SDP, we do not expect to surpass Average Linkage in the restrictive distributed context.

After the completion of our work, a new algorithm was introduced achieving an impressive 0.585 approximation for revenue~\cite{revapx}. While this result was unknown during the course of this work, it sets a new standard to strive for in future work.

\subsection{Our contributions}\label{sec:contributions}
In this work, we propose a new algorithm, \textit{Matching Affinity Clustering}, for distributed hierarchical clustering. Inspired by Affinity Clustering's reliance on the minimum spanning tree in order to greedily merge clusters \citep{affinity}, Matching Affinity Clustering merges clusters based on iterative matchings. It notably generalizes to both the edge weight similarity and dissimilarity contexts, and achieves all four desired qualities.

In Section~\ref{sec:results}, we \textbf{theoretically motivate} Matching Affinity Clustering by proving it achieves a good approximation for both revenue and value (the latter depending on the existence of an MPC minimum matching algorithm), nearing the bounds achieved by Average Linkage:

\begin{theorem}\label{thm:matchcluster}
    \matchcluster
\end{theorem}

\begin{theorem}\label{thm:matchclusterval}
    \matchclusterval
\end{theorem}

Furthermore, in Theorem~\ref{thm:affinity}, we prove that Matching Affinity Clustering can give no guarantees with respect to revenue or value. The discussion and proof of this theorem can be found in the Appendix.

\begin{theorem}\label{thm:affinity}
\affinity
\end{theorem}

We also present an efficient and near-optimal MPC algorithm for $k$-sized maximum matching in Theorem~\ref{thm:kmatch} in Section~\ref{sec:kmatch}. This is used by Matching Affinity Clustering.


To evaluate the \textbf{empirical performance} of our algorithm, we run \citet{affinity}'s experiments used for Affinity Clustering on small-scale datasets in Section~\ref{sec:experiments}. We find Matching Affinity Clustering performs competitively with respect to state-of-the-art algorithms. On filtered, balanced data, we find that Matching Affinity Clustering consistently outperforms other algorithms by at least a small but clear margin. This implies Matching Affinity Clustering may be more useful on balanced datasets than Affinity Clutsering.

To confirm the \textbf{balance} of our algorithm, we are able to prove that Matching Affinity Clustering achieves perfectly balanced clusters on datasets of size $2^N$, and otherwise guarantee near balance (a cluster size ratio of at most 2). See Lemma~\ref{lem:half}. This was also confirmed in our empirical evaluation in Section~\ref{sec:experiments}.

Finally, we show in Section~\ref{sec:results} that Matching Affinity Clustering is highly \textbf{scalable} because it was designed in the same MPC framework as Affinity Clustering. We provide similar complexity guarantees to Affinity Clustering.

Matching Affinity Clustering is ultimately a nice, simply motivated successor to Affinity Clustering that achieves all four desired qualities: empirical performance, theoretical guarantees, balance, and scalability. No other algorithm that we know of does this.

\section{Background}
In this section, we describe basic notation, hierarchical clustering cost functions, and Massively Parallel Communication (MPC). 


\subsection{Preliminaries}\label{sec:prelim}
The standard hierarchical clustering problem takes in a set of data represented as a graph where weights on edges measure similarity or dissimilarity between data. In this paper, edge weights, denoted $w_G(u,v)$ for a graph $G$ and may be similarities or differences as specified.


One of the most simple, serial hierarchical clustering algorithms is Average Linkage, which provides a good approximation for Moseley and Wang's dual \citep{dual}. It starts with each vertex being its own cluster, and builds the next cluster by merging the two clusters $A$ and $B$ that maximize the average distance between points in the clusters:
\[\frac{1}{|A|\cdot |B|} \sum_{u,v\in G} w_G(u,v).\]

Finally, we note that much of our results occur \textit{with high probability} (whp), as is true with many MPC algorithms. This means they occur with probability $1-1/(\Omega(n))$, where the denominator is generally exponential in $n$. The probabilistic aspects of the algorithm come from our use of \citet{matching}'s maximum matching algorithm, which finds a $(1+\epsilon)$-approximate maximum matching with high probability. As we do not introduce a probabilistic aspect ourselves, this will not be discussed in depth in our proofs, but will be stated in the theorems and lemmas.

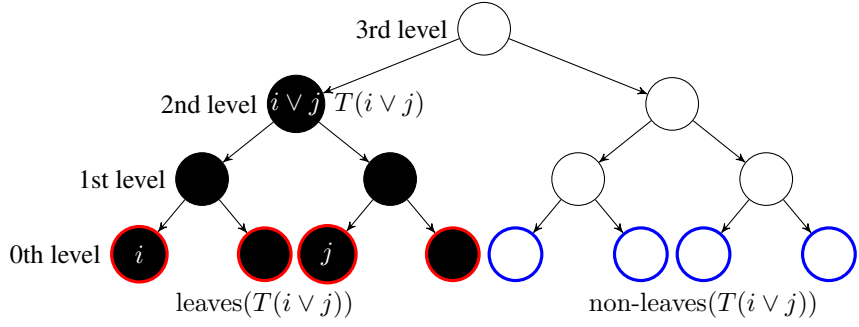
\begin{figure}
  \centering
  \begin{tikzpicture}[->,>=stealth',level/.style={sibling distance = 5cm/#1,
  level distance = 1cm}] 
\node (t1) [arn_bwb,label=left:3rd level] {}
    child{ node [arn_wbb,label=left:2nd level,label=right:$T(i\lor j)$] {$i\lor j$} 
            child{ node (t2) [arn_wbb,label=left:1st level] {} 
            	child{ node (t3) [arn_wbr,label=left:0th level] {$i$} 
                         } 
							child{ node (t4) [arn_wbr,label=below:$\leaves(T(i\lor j))$] {} }
            }
            child{ node (t5) [arn_wbb] {}
							child{ node (t6) [arn_wbr] {$j$}}
							child{ node (t7) [arn_wbr] {}}
            }                            
    }
    child{ node [arn_bwb] {}
            child{ node [arn_bwb] {} 
							child{ node [arn_bl] {}}
							child{ node [arn_bl] {}}
            }
            child{ node [arn_bwb] {}
							child{ node [arn_bl,label=below:$\nonleaves(T(i\lor j))$] {}}
							child{ node [arn_bl] {}}
            }
		}
		
; 
\end{tikzpicture}
  \caption{An example hierarchical tree $T$. $i$ and $j$ are data points, or leaves in the tree. Then $T(i\lor j)$, the subtree in black, is rooted at their least common ancestor. The red vertices are the leaves of this true, and the blue are the non-leaves. }
  \label{fig:tree}
\end{figure}

\subsection{Optimization functions}\label{sec:cost}
 Consider some hierarchical tree, $T$. We say $i\lor j$ for leaves $i$ and $j$ is the least common ancestor of $i$ and $j$. The subtree rooted at an interior vertex $v$ is $T[v]$, therefore the subtree representing the smallest cluster that contains both $i$ and $j$ is $T[i\lor j]$. Let $\leaves(T[v])$ be the set of leaves in $T[v]$, and $\nonleaves(T[v])$ be the set of all of the leaves of $T$ but not $T[v]$. Now we can describe Dasgupta's function.

\begin{definition}[\citet{dasgupta}]
Dasgupta's \textbf{cost} function of tree $T$ on graph $G$ with similarity-based edge weights $w_G$ is a minimization function.
\[\cost_G(T) = \sum_{i,j\in V(G)} w_G(i,j) |\leaves (T[i\lor j])|.\]
\end{definition}

To minimize edge weight contribution, we want a small $|\leaves (T[i\lor j])|$ for heavy edges. This ensures that heavy edges will be merged earlier in the tree. To calculate this, it is easier to break it down into a series of merge costs for each node in $T$. It counts the costs that accrue due to the merge at that node so that we can keep track of the cost throughout the construction of $T$. It is defined as:

\begin{definition}[\citet{dual}] 
The \textbf{merge cost} of a node in $T$ which merges disjoint clusters $A$ and $B$ is:
\begin{align*}
\mergecost_G(A, B) =& |B|\sum_{a\in A, c\in G\setminus(A\cup B)}w_G(a,c) +|A|\sum_{b\in B,c\in G\setminus(A\cup B)} w_G(b,c).
\end{align*}
\end{definition}

This breaks down the cost of a hierarchy tree into a series of merge costs. Consider some edge, $(i,j)$. At each merge containing exclusively $i$ or $j$, this edge contributes $w_G(i,j)$ times the size of the other cluster. In the hierarchical tree, this counts how many vertices accrue during merges along the paths from $i$ and $j$ to $i\lor j$. However, this does not account for the leaves $i$ or $j$ themselves, so we need to add $w_G(i,j)$ two extra times in addition to each merge. This means we can derive the total cost from the merge costs as: $\cost_G(T) = 2\sum_{i,j\in V(G)} w_G(i,j) + \sum_{\text{merges } A,B} \mergecost_G(A,B).$

Next, we consider Moseley and Wang's dual to Dasgupta's function \cite{dual}.

\begin{definition}[\citet{dual}] 
The \textbf{revenue} of tree $T$ on graph $G$ with similarity-based edge weights is a maximization function.
\[\rev_G(T) = \sum_{i,j\in V(G)} w_G(i,j) |\nonleaves (T[i\lor j])|.\]
\end{definition}

We can, in a similar fashion to Dasgupta's cost function, break revenue down into a series of merge revenues.

\begin{definition}\cite{dual}. 
The \textbf{merge revenue} of a node in $T$ which merges disjoint clusters $A$ and $B$ is:
\[\mergerev_G(A,B) = (n-|A|-|B|)\sum_{a\in A,b\in B} w_G(a,b).\]
\end{definition}

Note that for some $i$ and $j$, $w_G(i,j)$ is contributed exactly once, when $i$ and $j$ merge at $i\lor j$, and $n-|A|-|B|$ is the number of non-leaves at that step. Therefore: $\rev_G(T) = \sum_{i,j\in V(G)} \mergerev_G(i,j)$.
 In addition, note the contribution of each $i,j$ pair, which is scaled by $w_G(i,j)$, is the number of leaves of $i\lor j$ for revenue, and the number of non-leaves of $i\lor j$ for cost. Therefore the contribution of each edge for revenue is $n$ minus the contribution for cost, scaled by $w_G(i,j)$. In other words: $\rev_G(T) = n\sum_{i,j\in V(G)} w_G(i,j) - \cost_G(T)$.

While cost is a popular and well-founded metric, \citet{dasguptabounds} found that it is not constant factor approximable under the Small Set Expansion Hypothesis. On the other hand,~\citet{dual} proved that revenue is, and Average Linkage achieves a $1/3$-approximation. This makes it a more practical function to work with.

Our other function of interest is~\citet{cohen-addad}'s value function. This was introduced as a Dasgupta-inspired optimization function where edge weights represent distances. It looks exactly like cost, except it is now a maximization function because it is in the distance context.

\begin{definition}[\citet{cohen-addad}] 
The \textbf{value} of tree $T$ on graph $G$ with dissimilarity-based edge weights $w_G$ is a maximization function.
\[\val_G(T) = \sum_{i,j\in V(G)} w_G(i,j) |\leaves (T[i\lor j])|.\]
\end{definition}

Like revenue, value is constant factor-approximable. In fact, the best approximation (other than an SDP) for value is Average Linkage's $2/3$-approximation~\citep{cohen-addad}. To our knowledge, there are no distributable approximations for value.

\subsection{Massively Parallel Communication (MPC)}\label{sec:MPC}

Massively Parallel Communication (MPC) is a model of distributed computation used in programmatic frameworks like MapReduce \citep{mapreduce2}, Hadoop \citep{hadoop}, and Spark \citep{spark}. MPC consists of ``rounds'' of computation, where parts of the input are distributed across machines with limited memory, computation is done locally for each machine, and then the machines send limited messages to each other. The primary complexities of interest are machine space, which should be $\widetilde{O}(n)$, and the number of rounds. Many MPC algorithms are extremely efficient. For instance, Affinity Clustering in some cases can have constantly many rounds, and otherwise may use up to $O(\log^2n)$ rounds \citep{affinity}.

\section{Finding a k-sized maximum matching}\label{sec:kmatch}

The algorithm we introduce in Section \ref{sec:results} requires the use of a $(1-\epsilon)$-approximation for the maximum $k$-sized (or less) matching, where $k> n/2$. For this we will use \citet{matching}'s $(1-\epsilon)$-approximation for maximum matching in MPC, which runs in $O(\log\log(n)\cdot (1/\epsilon)^{1/\epsilon})$ rounds with $O(n/polylog(n))$ space. Inspired by the results of~\citet{kmatching}, we provide a distributed reduction between matching and $k$-matching. To do this, we add $n-2k$ vertices and edges of weight $Q$ (which is found with a binary search) between the new and original vertices, and run the matching algorithm. The algorithm can be seen below and the proof is found in the Appendix in the full paper.

\begin{theorem}\label{thm:kmatch}
\kmatch
\end{theorem}

\begin{algorithm}[b]
   \caption{Approximate $k$-Sized Matching}
   \label{alg:kmatch}
\begin{algorithmic}[1]
\STATE  Let $U$ be set of dummy vertices for $|U|=n-2k$\;\label{line:transstart}
\STATE  Let $\delta$ be the minimum of $\epsilon$ and the value satisfying $\epsilon = (c+1)\delta - \delta^2$ for $k \leq cn$\;
\STATE  $V'\gets V\cup U$\tcp*{Constructing the transformed graph}
\STATE  $E'\gets E\cup \{(u,v): u\in U, v\in V\}$\;
\WHILE{Binary search of $Q\in[nW]$ for the min $Q$ that results in $|M|\leq k$ and $\frac1{k(1-\delta)}w(M) \leq Q$}\label{line:binary}
\STATE    $w'(u,v) = Q$ for all $u\in U, v\in V$\;\label{line:transend}
\STATE    $M \gets \textsc{Match}(V',E')$\tcp*{\citet{matching}'s matching algorithm}\label{line:ghaf}
\STATE$M \gets M \setminus \{(u,v):(u,v)\in M, u\in U, v\in V\}$\tcp*{Remove edges not in $G$}
\ENDWHILE
\end{algorithmic}
\end{algorithm}

\section{Bounds on a matching-based hierarchical clustering algorithm}\label{sec:results}

We now introduce our main algorithm, \textit{Matching Affinity Clustering}. For revenue, we show it achieves a $\left(\frac13-\epsilon\right)$-approximation for graphs with $2^N$ vertices, and a $\left(\frac19-\epsilon\right)$-approximation in general.  Similarly, for value, we show it achieves a $\frac23\alpha$-approximation for graphs with $2^N$ vertices, and a $\frac13\alpha$-approximation in general, given an $\alpha$-approximation algorithm for minimum weighted $k$-sized matching in MPC.

\subsection{Matching Affinity Clustering}

\begin{figure}[b]
  \centering
  \vspace{0cm}
  \begin{tikzpicture}[node distance=1.3cm,>=stealth',bend angle=45,auto,scale=0.8]

  \tikzstyle{place}=[circle,thick,draw=blue!75,fill=blue!20,minimum size=6mm]
  \tikzstyle{red place}=[place,draw=red!75,fill=red!20]
  \tikzstyle{transition}=[rectangle,thick,draw=black!75,
  			  fill=black!20,minimum size=4mm]

  \tikzstyle{every label}=[red]
  \begin{scope}
    \node[arn_bwbb] (A) at (0,-1.5) {$a$};
    \node[arn_bwb] (B) at (0,-2.5) {$b$};
    \node[arn_bwbb] (C) at (2.5,-3) {$c$};
    \node[arn_bwbb] (D) at (2.5,-1.75) {$d$};
    \node[arn_bwbb] (E) at (3.5,-1) {$e$};
    \node[arn_bwbb] (F) at (4,-2.5) {$f$} ;
    \node[arn_bwbb] (G) at (1, -0.5) {$g$};

    \path [-](A) edge node {} (B);
    \path [-](B) edge node {} (C);
    \path [-](A) edge node {} (D);
    \path [-](D) edge node {} (C);
    \path [-](A) edge node {} (E);
    \path [-, very thick](D) edge node {} (E);
    \path [-](D) edge node {} (F);
    \path [-, very thick](C) edge node {} (F);
    \path [-](E) edge node {} (F);
    \path [-, very thick](A) edge node {} (G);
    \path [-](E) edge node {} (G);  
  \end{scope}

  \begin{scope}[xshift=6cm]
     \node[arn_bwb] (AG) at (0.5,-1) {$a,g$};
    \node[arn_bwb] (B) at (0,-2.5) {$b,b'$};
    \node[arn_bwb] (CF) at (3.25,-2.75) {$c,f$};
    \node[arn_bwb] (DE) at (3,-1.375) {$d,e$};

    \path [-](AG) edge node {} (B);
    \path [-](B) edge node {} (CF);
    \path [-](AG) edge node {} (DE);
    \path [-](DE) edge node {} (CF);
    \path [-](DE) edge node {} (AG);  
  \end{scope}
\end{tikzpicture}
  \caption{An example of the first iteration of Matching Affinity Clustering. We start by doing a 6-sized matching on the current graph on the right. We then duplicate unmatched vertices and merge to create the next cluster graph with $2^n$ vertices on the right. In subsequent iterations, matches are perfect. Edge weights are the Average Linkage between clusters (non-edges are zero).}
  \label{fig:alg}
  \vspace{-0.5cm}
\end{figure}
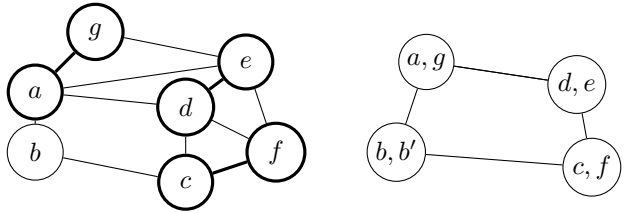

Matching Affinity Clustering is defined in Algorithm \ref{alg:matchcluster}. Its predecessor, Affinity Clustering, uses the MST to select edges to merge across, which sometimes causes imbalanced clusters. This is one reason why it cannot achieve a good approximation for revenue or value (Theorem \ref{thm:affinity}). We fix this by, instead, using iterated maximum matchings (for similarity edge weights) and minimum perfect matchings (for dissimilarity edge weights). This ensures that on $n = 2^N$ vertices for some $N$, clusters will always be perfectly balanced. Otherwise, we will show how to achieve at least relative balance.

The algorithm starts with one cluster per each of $n$ vertices. Let $2^N$ be the smallest value such that $2^N\geq n$. It finds a maximum (resp. minimum) matching of size $k= 2n-2^N$ (line 8, this means it matches $2n-2^N$ vertices with $n-2^{N-1}$ edges) and merges these vertices (line 12, Figure~\ref{fig:alg}). Note that if $n=2^N$, then $k=0$ and the first step is a perfect matching. After this step, we have $2^{N-1}$ clusters. We then transform the graph into a graph of clusters with edge weights equal to the Average Linkage between clusters (lines 17-21). We find a maximum (resp. minimum) perfect matching of clusters in this new graph (line 10), then iterate.

\begin{algorithm}[tb]
   \caption{Matching Affinity Clustering}
   \label{alg:matchcluster}
\begin{algorithmic}[1]
\STATE {\bfseries Input:} A graph $G$ with weight function $w:E(G)\to \mathbb{Z}^+$.
\STATE $n\gets |V|$
\STATE $N \gets$ Such that $2^{N-1} < n \leq 2^N$\label{line:2n}
\STATE $\mathcal{C} \gets G $\tcp*{Current clustering graph, see Definition \ref{def:clustgraph}}
\WHILE{$n > 1$}
\STATE  Yield $\mathcal{C}$\tcp*{Output each level of the hierarchy}
\IF{First iteration}
\STATE  \label{line:startmerge} $M \gets \textsc{kMatch}(\mathcal{C},2n-2^{N})$\tcp*{Alg. 2 (Appendix)}
\ELSE
\STATE    $M \gets \textsc{Match}(\mathcal{C})$\tcp*{\cite{matching}}\label{line:merge}
\ENDIF
\STATE   $V \gets \{v = (i,j) : (i,j)\in M\}$\;\label{line:duplicate}
\STATE  $E\gets V\times V$\;
\STATE  $w \gets \emptyset$\;
\STATE  $n\gets |V|$\;
\STATE  Allocate each $C_j \in V$ to a machine\;
\FOR{Every machine $m_j$ on $C_j$ that merged $A_j, B_j \in V(\mathcal{C})$}\label{line:transform}
\FOR{Every other $C_k\in V$ that merged $A_k, B_k\in V(\mathcal{C})$}
\STATE $w(C_j, C_k) \gets \frac14(w_{\mathcal{C}}(A_j,A_k) + w_{\mathcal{C}}(A_j,B_k) + w_{\mathcal{C}}(B_j,A_k) + w_{\mathcal{C}}(B_j,B_k))$
\ENDFOR
\ENDFOR
\STATE $\mathcal{C} \gets (V,E,w)$
\ENDWHILE
\end{algorithmic}
\end{algorithm}

\subsection{Revenue approximation}\label{sec:const}

Now, we evaluate the efficiency and approximation factor of Matching Affinity Clustering with respect to revenue. In this section, edge weights represent the \textit{similarity} between points. Proofs are in the Appendix in the full version of the paper. Ultimately, we will show the following.

\begin{customthm}{\ref{thm:matchcluster}}
    \matchcluster
\end{customthm}

This will be a significant motivation for Matching Affinity Clustering's theoretical strength. As stated previously, one of the goals of Matching Affinity Clustering is to keep the cluster sizes balanced at each level. However, in the first step, note that Matching Affinity Clustering creates $n-2^{N-1}$ clusters of size 2, and the rest of the vertices form singleton clusters. Therefore, to use this benefit of Matching Affinity Clustering, we need to ensure that cluster sizes will never deviate \textit{too much}. 

\begin{lemma}\label{lem:half}
\half
\end{lemma}

After every matching, the algorithm creates a new graph with vertices representing clusters and edges representing the average linkage between clusters. We will call this a clustering graph.

\begin{definition}\label{def:clustgraph}
    A \textbf{clustering graph} $\mathcal{C}(G,C)$ for graph $G$ and clustering $C = \{C_1,\ldots,C_k\}$ of $G$ is a complete graph with vertex set $V = C$. Its edge weights are the average linkage between clusters. Specifically, for vertices $v_{C_i}$ and $v_{C_j}$ in $\mathcal{C}(G,C)$ corresponding to clusters $C_i$ and $C_j$ where $i\neq j$, the weight of the edge between these vertices is:
    \[w_{\mathcal{C}(G,C)}(v_{C_i},v_{C_j}) = \frac1{|C_i| \cdot |C_j|}\sum_{u\in C_i, w\in C_j} w_G(u,w).\]
\end{definition} 

The fact that the edge weights in the clustering graph are the average linkage between clusters denotes the similarities between Matching Affinity Clustering and Average Linkage. Essentially, we are trying to optimize for average linkage at each step, but instead of merging two clusters, we merge many pairs of clusters at once with a maximum matching.
 
 Since Matching Affinity Clustering computes this graph, we must show how to efficiently transform a clustering graph at the $i$th level, $\mathcal{C}(G,C^i)$ with clustering $C^i$, into a clustering at the $i+1$th level, $\mathcal{C}(G,C^{i+1})$ with clustering $C^{i+1}$.

\begin{lemma}\label{lem:transform}
\transform
\end{lemma}

This will eventually be used for our proof of efficiency of Theorem \ref{thm:matchcluster}. For now, we return our attention to the approximation factor. Our approximation proof is going to observe the total merge cost and revenue across all merges on a single level of the hierarchy. For concision, we introduce the following notation to describe cost and revenue over a single clustering.

\begin{definition}
    The \textbf{clustering revenue} based off of some superclustering $C'$ of $C$ on graph $G$ is the sum of the merge revenues of combining clusters in $C$ to create clusters in $C'$. It is denoted by $\clusterrev_G(C, C')$.
\end{definition}
\begin{definition}
    The \textbf{clustering cost} based off of some superclustering $C'$ of $C$ on graph $G$ is the sum of the merge costs of combining clusters in $C$ to create clusters in $C'$. It is denoted by $\clustercost_G(C,C')$.
\end{definition}

In order to prove an approximation for revenue, we want to compare each clustering revenue and cost. First, we must show that Matching Affinity Clustering has a large clustering revenue at any level.

\begin{lemma}\label{lem:rev}
\lemrev
\end{lemma}

Now we address clustering cost. This time, we must show an upper bound for clustering cost at the $i$th level in terms of clustering revenue at the $i$th level. Let $M_i$ be the matching Matching Affinity Clustering uses to merge $C^i$ into $C^{i+1}$. Then $M_i$ is a $(1-\epsilon)$-approximation of the optimum $M_i^*$.

\begin{lemma}\label{lem:costb}
\costb
\end{lemma}

Now we are ready to prove the approximation factor for Matching Affinity Clustering. We combine Lemma \ref{lem:costb} with properties of revenue from Section~\ref{sec:cost} to obtain an expression for revenue in terms of $(n-2)$ times the sum of weights in the graph. We use this as a bound for the optimal revenue.

\begin{lemma}\label{lem:apx}
\lemapx
\end{lemma}

Finally, the round complexity is limited by the iterations and calls to the matching algorithm. The space complexity is determined by the clustering graph construction.

\begin{lemma}\label{lem:bounds}
\bounds
\end{lemma}

Lemmas \ref{lem:apx} and \ref{lem:bounds} are sufficient to prove Theorem \ref{thm:matchcluster}. Our algorithm achieves an approximation for revenue efficiently in the MPC model. In addition, the algorithm creates a desirably near-balanced hierarchical clustering tree.


We now prove the approximation bound tightness for Matching Affinity Clustering when $|V|=2^N$. Recently, \citet{linkage} proved by counterexample that Average Linkage achieves at best a $(1/3 + o(1))$-approximation on certain graphs. We find that Matching Affinity Clustering acts the same as Average Linkage on these graphs, and so has at best a $(1/3 + o(1))$-approximation.


\begin{theorem}\label{thm:lb}
\lb
\end{theorem}

\subsection{Value approximation}
Now we consider Matching Affinity Clustering when edge weights represent \textit{distances} instead of similarities. In this context, instead of running a $k$-sized maximum matching and then iterative general maximum matchings, we run a $k$-sized minimum matching and then iterative general minimum perfect matchings. Therefore, this algorithm is dependent on the existence of a $k$-sized minimum matching algorithm in MPC. Due to its similarity to other classical problems with $1+\epsilon$ solutions in MPC~\citep{matching, maximalmatch}, we conjecture:

\begin{conjecture}\label{claim}
There exists an MPC algorithm that achieves a $(1+\epsilon)$-approximation for minimum weight $k$-sized matching whp that uses $\widetilde{O}(n)$ machine space.
\end{conjecture}

Given such an algorithm, we can show that Matching Affinity Clustering approximates value.

\begin{theorem}\label{thm:val}
\matchclusterval
\end{theorem}

The proof for this result is quite similar to the proof for the $2/3$-approximation of Average Linkage by \citet{cohen-addad}. Instead of focusing on single merges, however, we observed the entire set of merges across a clustering layer in our hierarchy. Then we can make the same argument about the value across an entire level of the hierarchy, and use the cluster balance from Lemma~\ref{lem:half} to achieve our result.

If Conjecture~\ref{claim} holds, then the approximation factors become $2/3-\epsilon$ and $1/3-\epsilon$ respectively. We see a similar pattern as the revenue result, where the algorithm nears the state-of-the-art $2/3$-approximation achieved by Average Linkage \citep{cohen-addad} on datasets of size $n=2^N$, and still achieves a constant factor in general. Finally, we can additionally show the former approximation is tight. See the construction and proofs in the Appendix.

\begin{theorem}\label{thm:lbval}
\lbval
\end{theorem}

\subsection{Round comparison to Affinity Clustering}
In this section, we only consider Matching Affinity Clustering in the similarity edge weight context. The round complexities of Matching Affinity Clustering and regular Affinity Clustering depend on graph qualities, and in certain cases one outperforms the other. On dense graphs with $n^{1+c}$ edges for constant $c$, \citet{affinity} showed that Affinity Clustering runs in $\lceil\log(c/\epsilon)\rceil + 1$ rounds. On sparse graphs, it runs in $O(\log^2n)$ rounds, and it runs in $O(\log n)$ rounds when given access to a distributed hash table. We saw that Matching Affinity Clustering runs in $O\left(\log(n)\log\log(n)\cdot(1/\epsilon)^{O(1/\epsilon)}\right)$ rounds on graphs of size $2^N$, and $O\left(\log(nW)\log\log(n)\cdot(1/\epsilon)^{O(1/\epsilon)}\right)$ in general for max edge weight $W$.

There are two situations where our algorithm outperforms Affinity Clustering. First, if the graph is sparse and the number of vertices is $2^N$, then our algorithm runs in $O\big(\log(n)\allowbreak\log\log(n)\cdot (1/\epsilon)^{O(1/\epsilon)}\big)$ rounds, and Affinity Clustering runs in $O(\log^2(n))$ rounds. Otherwise, if the graph is sparse, Matching Affinity Clustering performs better as long as the largest edge weight is $W = o\left(\frac{\exp(\log^2(n)/\log\log(n))}{n}\right)$. This is strictly larger than constant. If $W$ is large, Affinity Clustering is slightly more efficient. Finally, if the graph is dense, Affinity Clustering achieves an impressive constant round complexity, and is therefore more efficient. In any case, Matching Affinity Clustering is an efficient and highly scalable algorithm.
\section{Affinity Clustering approximation bounds}

In this section and all following sections, we provide the proofs for all theorems and lemmas introduced in this paper. It is broken down into sections based off of the sections corresponding to sections in the paper itself. 
 
We start by proving Theorem~\ref{thm:affinity}. \citet{affinity} were in part motivated by the lack of theoretical guarantees for distributed hierarchical clustering algorithms. Thus, they introduced Affinity Clustering, based off of \citet{boruvka}'s algorithm for parallel MST. In every parallel round of Bor{\r u}vka's algorithm, each connected component (starting with disconnected vertices) selects the lowest-weight outgoing edge and adds that to the solution, eventually creating an MST. Affinity Clustering creates clusters of each component. Note that Affinity Clustering was evaluated on a graph with weights representing dissimilarities between vertices, as opposed to our representation where weights are similarities. It is easy to verify that Affinity Clustering functions equivalently using max spanning tree in our representation. \citet{affinity} theoretically validate their algorithm by defining a cost function based off the cost of the minimum Steiner tree for each cluster in the hierarchy, however they do not motivate this metric. Therefore, it is more interesting to evaluate in terms of revenue and value. We ultimately show:

\begingroup
\def\thetheorem{\ref{thm:affinity}}
\begin{theorem}
\affinity
\end{theorem}
\addtocounter{theorem}{-1}
\endgroup 


We will split this into two cases for each objective function. We start with revenue. First, note that when Affinity Clustering merges clusters in common connected components, it creates one supercluster (ie, cluster of clusters) for all clusters in that component. Therefore, it may merge many clusters at once. A brief counterexample of why such a hierarchy does not work is when the max spanning tree is a star. Here, all vertices will be merged to a cluster in one round for a revenue of zero, which is not approximately optimal. To evaluate this algorithm, we must consider all possible ways Affinity Clustering might decide to resolve edges on the max spanning tree of the input graph. We propose a graph family that shows Affinity Clustering cannot achieve a good revenue approximation. The hierarchy we use for comparison is one that Matching Affinity Clustering would find, not including the $k$-matching step. We prove the following lemma.

\begin{lemma}\label{lem:affinityrev}
\affinityrev
\end{lemma}

We now move on to value.

\begin{lemma}\label{lem:affinityval}
\affinityval
\end{lemma}

Finally, we simply combine the results of Lemma~\ref{lem:affinityrev} and Lemma~\ref{lem:affinityval} to prove Theorem~\ref{thm:affinity}.

\section{Experiments}\label{sec:experiments}
We now empirically validate these results to further motivate Matching Affinity Clustering. The algorithm is implemented as a sequence of maximum or minimum perfect matchings, and the testing software is provided in supplementary material. The software as well as the five UCI datasets \cite{uci} ranging between 150 and 5620 data points are exactly the same as those that were used for small-scale evaluation of Affinity Clustering \cite{affinity}. The data is represented by a vector of features. Similarity-based edge weights are the cosine similarity between vectors, and dissimilarity-based edge weights are the L2 norm. Most data and algorithms are deterministic and thus have consistent outcomes, but for any randomness, we run the experiment 50 times and take the average. Just like the evaluation of Affinity Clustering, our evaluation runs hierarchical clustering algorithms on $k$-clustering problems until we find a $k$-clustering within the hierarchy. This was compared to the ground truth clustering for the dataset.

We evaluate performance using the Rand index, which was designed by \citet{rand} to be similar to accuracy in the unsupervised context. This is an established and commonly used metric for evaluating clustering algorithms and was used in the evaluation of Affinity Clustering.

\begin{definition}[\citet{rand}]\label{def:rand}
Given a set $V = \{v_1,\ldots,v_n\}$ of $n$ points and two clusterings $X = \{X_1,\ldots,X_r\}$ and $Y = \{Y_1,\ldots,Y_s\}$ of $V$, we define:
\begin{itemize}
\item $a$: the number of pairs in $V$ that are in the same cluster in $X$ and in the same cluster in $Y$.
\item  $b$: the number of pairs in $V$ that are in different clusters in $X$ and in different clusters in $Y$.
\end{itemize}

\begin{figure*}
\centering 
\subfigure[Rand Index on Raw Data ]{\label{fig:sub1}\includegraphics[width=80mm]{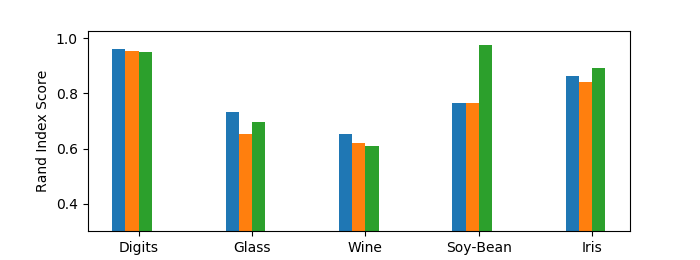}}
\subfigure[Rand Index on Filtered Data]{\label{fig:sub2}\includegraphics[width=80mm]{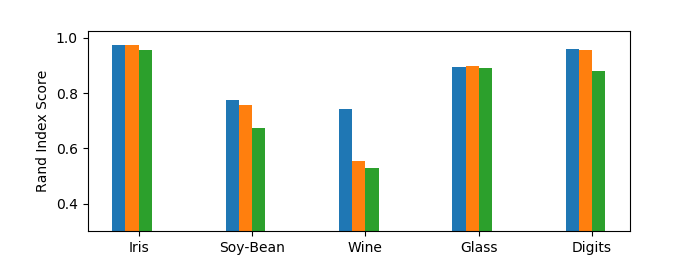}}
\subfigure[Cluster Balance on Raw Data]{\label{fig:sub3}\includegraphics[width=80mm]{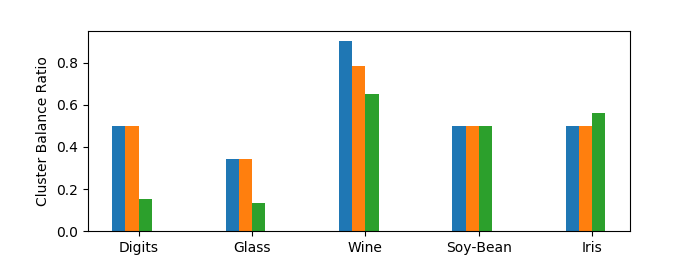}}
\subfigure[Cluster Balance on Filtered Data]{\label{fig:sub4}\includegraphics[width=80mm]{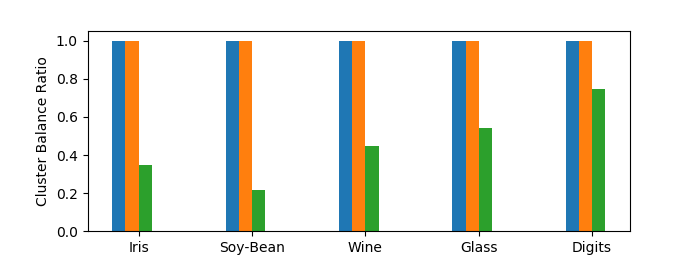}}
\caption{Rand Index and cluster balance on raw and filtered (randomly pruned for balance and $n=2^N$) UCI datasets. Legend (bars, left to right): Max Matching Affinity Clustering is blue, Min Matching Affinity Clustering is orange, Affinity Clustering is green.}
\end{figure*}

The Rand index $r(X,Y)$ is $(a+b)/\binom{n}{2}$. By having the ground truth clustering $T$ of a dataset, we define the Rand index score of a clustering $X$ to be $r(X,T)$.
\end{definition}

\subsection{Comparison with Affinity Clustering}

In addition, we are interested in evaluating the balance between cluster sizes in the clusterings, which indicates how good our algorithms are at evaluating balanced data. We use the \textit{cluster size ratio} of a clustering, which was observed in \citet{affinity}. For a clustering $X=\{X_1,\ldots,X_r\}$, the size ratio is $\min_{i,j\in[r]} |X_i|/|X_j|$.

In Figure \ref{fig:sub1}, we see the Rand indices of Max Matching Affinity Clustering (ie, in the similarity context), Min Matching Affinity Clustering (ie, in the distance context), and Affinity Clustering. A Rand index is between 0 and 1, where higher scores indicate the clustering is more similar to the ground truth. Matching Affinity Clustering performs similarly to state of the art algorithms like Affinity Clustering on all data except the Soy-Bean dataset. A full evaluation on other algorithms (see the Appendix) illustrates that Matching Affinity Clustering outperforms other algorithms like Random Clustering and Average Linkage.

Figure \ref{fig:sub2} depicts the same information but on a slightly modified dataset. Here, we randomly remove data until (1) the dataset is of size $2^N$, and (2) ground truth clusters are balanced. We did this 50 times and took the average results. This is motivated by Matching Affinity Clustering's stronger theoretical guarantees on datasets of size $2^N$ and ensured cluster balance. As expected, Matching Affinity Clustering performs consistently better than Affinity Clustering on filtered data, albeit by a a small margin in many cases. This shows that, experimentally, Matching Affinity Clustering performs better on balanced datasets of size $2^N$.

Finally, Figures \ref{fig:sub3} and \ref{fig:sub4} depict the cluster size ratios on the raw and filtered data respectively. In theory, at every level in the hierarchy of Matching Affinity Clustering, no cluster can be less than half as small as another (Lemma~\ref{lem:half}). However, in our evaluation, we are comparing a single $k$-clustering, which may not precisely correspond to a level in the hierarchy. In this case, we take some clusters from the first level with fewer than $k$ clusters and the last level with more than $k$ clusters. Therefore, since cluster sizes double at each level, the lower bound for the cluster size ratio is now $1/4$. This is reflected Figure~\ref{fig:sub3}, where Matching Affinity Clustering stays consistently above this minimum, and often exceeds it by quite a bit. On the filtered data (Figure~\ref{fig:sub4}), Matching Affinity Clustering maintains perfect balance in every instance, whereas Affinity Clustering performs much worse. Thus, Matching Affinity Clustering has proven empirically successful on small datasets.

\subsection{Comparison with other algorithms}
Here we include more complete visualizations of the performance of all tested algorithms. Like in the main body of text, we find the rand index and cluster size ratio on balanced and filtered data. These tests were run in the same way as the tests in the main body, we just present results on other common algorithms for completeness. The results are presented in Figure~\ref{fig:a} in the Appendix.

Most of these results are as expected and simply reproduce the results from \citet{affinity}. However, we add one more algorithm: random clustering. Again, this is the clustering that randomly recursively partitions the data into a hierarchy. In our experiments, random clustering had surprisingly good cluster balance ratios (see Figure~\ref{fig:suba3}). In fact, on raw data, it on average had more balanced clusters than Matching Affinity Clustering on three of the datasets.

There are three main observations about why Matching Affinity Clustering is still clearly more empirically successful than random clustering. First, notice that on filtered data in Figure~\ref{fig:suba4}, Matching Affinity Clustering has more balanced clusters than Random clustering by a very wide margin. Second, it is clear in Figures~\ref{fig:suba1} and~\ref{fig:suba2} that Matching Affinity Clustering consistently and significantly outperforms random clustering. Third, random clustering is nondeterministic, whereas Matching Affinity Clustering is works with high probability. Therefore, Matching Affinity Clustering's theoretical strengths and empirical performances are much stronger assurances than that of random clustering. Therefore, while an argument can be made that random clustering seems to empirically balance clusters well, Matching Affinity Clustering still does better in a number of respects, and thus is a more useful algorithm.

\section{Conclusion}\label{sec:conclusion}

Matching Affinity Clustering is the first hierarchical clustering algorithm to simultaneously achieve our four desirable traits. (1) Theoretically, it guarantees state-of-the-art approximations for revenue and value (given an approximation for MPC minimum perfect matching) when $n=2^N$, and good approximations for revenue and value in general. Affinity Clustering cannot approximate either function. (2) Compared to Affinity Clustering, our algorithm achieves similar empirical success on general datasets and performs even better when datasets are balanced and of size $2^N$. (3) Clusters are theoretically and empirically balanced. (4) It is scalable.

These attributes were proved through theoretical analysis and small-scale evaluation. While we were unable to perform the same large-scale tests as~\citet{affinity}, our methods still establish several advantages to the proposed approach. Matching Affinity Clustering simultaneously attains stronger broad theoretical guarantees, scalability through distribution, and small-scale empirical success. Therefore, we believe that Matching Affinity Clustering holds significant value over its predecessor as well as other state-of-the-art hierarchical clustering algorithms, particularly with its niche capability on balanced datasets.


\bibliography{references}
\bibliographystyle{icml2020}

\appendix

\section{Distributed maximum $k$-sized matching}
In this section, we prove our results for distributed $k$-sized matching and additionally provide the pseudocode.

\begingroup
\def\thetheorem{\ref{thm:kmatch}}
\begin{theorem}
\kmatch
\end{theorem}
\addtocounter{theorem}{-1}
\endgroup

\begin{proof}
Let us define $Q$ by a binary search on values 1 through $nW$ to find the minimum $Q$ that satisfies a halting condition: that the resulting $M$ the algorithm finds satisfies $|M|=k$ and $\frac1{k(1-\delta)}w(M) = Q$ (line \ref{line:binary}) where $W$ is the largest weight in $G$. First we transform the graph. Create a vertex set $U$ of $n-2k$ dummy vertices, add them to our vertex set, and connect them to all vertices in $G$ with edge weights $Q$ (lines \ref{line:transstart} to \ref{line:transend}). Then we run \citet{matching}'s algorithm (line \ref{line:ghaf}) on this new graph with error $\delta$ being the minimum of the value satisfying $\epsilon = (c+1)\delta - \delta^2$ and $\epsilon$ itself, and $k \leq cn$. Find the portion of this matching in $G$, and use this to check our halting condition.

We must start by showing that if $M_{G,k}$ is a $(1-\epsilon)$-approximate $k$-matching in $G$, then when $Q = \frac1{k(1-\epsilon)}w(M_{G,K})$, our algorithm finds a $k$-matching with a $(1-\epsilon)$-approximate weight. Assume there is such a matching, and consider when our algorithm selects $Q$ for this value. Consider the matching the algorithm finds in the transformed graph. Assume for contradiction that some $u\in U$ is not matched to any $v\in V$, and every edge in the matching from $G$ has weight over $Q$. Because $u$ is not connected to any vertex in $U$, that means it isn't matched at all. And since $u$ is connected to all $v\in V$ with a positive edge, for $u$ to not have been matched, all $v\in V$ must have been matched. Since $k < n/2$, a perfect matching on $G$ has at least $k$ edges. Thus the weight of the algorithm's matching in $G$, $M_{\mathcal{A},G}$, is bounded below by $k$ edges of weight greater than $Q$.

\[w(M_{\mathcal{A},G}) > kQ = \frac1{1-\epsilon}w(M_{G,k}).\]

But \textsc{OPT}$_{G,k}$ for $k$-sized matching in $G$ must be at least as large as this, so then $(1-\epsilon)w(\textsc{OPT}_{G,k}) > w(M_{G,k})$, which is a contradiction on the assumption of $M_{G,k}$. Otherwise, if there is some $u\in U$ that is not matched to any $v\in V$ where one of the edges in the matching from $G$ has weight $Q$ or less, removing that match and pairing one of those vertices with $u$ can only improve the matching. Thus we can add a processing step afterwards to ensure all $n-2k$ new vertices are matched, while not decreasing the value of the matching.

Thus our algorithm matches all $n-2k$ vertices in $U$ to $n-2k$ vertices in $V$, and the remaining $2k$ vertices in $V$ create a matching for a total size of $k$ or less. Thus the algorithm outputs an at most $k$-sized matching in $G$, so this selection of $Q$ will make the halting condition true.

Let $M_{\mathcal{A}}$ be the matching our algorithm finds in the transformed graph. Then the total weight is,

\[w(M_{\mathcal{A}}) = (n-k)Q + w(M_{\mathcal{A},G}).\]

By the same argument as before, but without the $1-\epsilon$ factor, there is an \textsc{OPT} in the transformed graph that matches all vertices in $U$ to vertices in $V$. Thus the expression for \textsc{OPT} is similar, where \textsc{OPT}$_{G,k}$ is the optimum for a $k$-sized matching in $G$.

\[w(\textsc{OPT}) = (n-k)Q + w(\textsc{OPT}_{G,k}).\]

We know $M_{\mathcal{A}}$ is a $(1-\epsilon)$-approximation for \textsc{OPT}, so we can combine these two equations.

\begin{align*}(n-k)Q + w(M_{\mathcal{A},G})\geq& (1-\epsilon)(n-k)Q + (1-\epsilon)w(\textsc{OPT}_{G,k}) .\end{align*}

We are interested in the portion of the solution in $G$, or $M_{\mathcal{A},G}$.

\[w(M_{\mathcal{A},G}) \geq -\epsilon(n-k)Q + (1-\epsilon)w(\textsc{OPT}_{G,k}).\]

Recall that $Q = \frac1{k(1-\epsilon)} w(M_{G,k})$, and $M_{G,k}$ is a $k$-sized matching in $G$. Therefore $w(M_{G,k}) \leq(\textsc{OPT}_{G,k})$. We can apply this to our inequality and simplify.

\begin{align*}
w(M_{\mathcal{A},G}) \geq& -\epsilon(n-k)\frac1{k(1-\epsilon)}w(\textsc{OPT}_{G,k}) + (1-\epsilon)w(\textsc{OPT}_{G,k}),
\\=& \left(1-\epsilon(1+(n/k - 1)(1-\epsilon))\right)w(\textsc{OPT}_{G,k}).
\end{align*}

Since $k=O(n)$, then $n/k$ is bounded above by some constant. Then the approximation factor is $1-\epsilon(1 + c(1-\epsilon)) = 1 - (c+1)\epsilon + \epsilon^2$. So for any approximation factor $\epsilon$, we can select some $\delta$ to run \citet{matching}'s algorithm such that our algorithm gives a $(1-\delta)$-approximation.

The algorithm searches for the minimum $Q$ that satisfies this, so all that's left to prove is that a selection of $Q < \frac1{k(1-\delta)} w(M_{G,k})$ yields a $k$-sized matching $M_{\mathcal{A},G}$ where $Q = \frac1{k(1-\delta)}w(M_{\mathcal{A},G})$ still is a $(1-\epsilon)$-approximation. If this is true, it must have matched all $n-2k$ vertices in $U$ with vertices in $V$ and selected $k$ edges from $G$. The value of this, where we sub in our value for $Q$, is:

\begin{align*}
w(M_{\mathcal{A}}) &= (n-2k)Q + w(M_{\mathcal{A},G}),
\\&= (n-2k)w(M_{\mathcal{A},G}) + w(M_{\mathcal{A},G}),
\\&= (n-2k+1)w(M_{\mathcal{A},G}).
\end{align*}

Therefore, the approximation factor on the transformed graph is equivalent to the approximation factor on $G$. Since we ran \citet{matching}'s algorithm on the transformed graph with error $\delta$ where $\delta \leq \epsilon$, this yields a matching within error of \textsc{OPT}$_{G,k}$.

    Therefore, our algorithm returns the desired approximation. This algorithm requires $O(\log(nW))$ iterations for the binary search on $Q$. In each iteration, the only significant computation in both round and space complexity is the use of the \citet{matching} algorithm that uses $O(\log\log(n)\cdot(1/\epsilon)^{1/\epsilon})$ rounds and $O(n/polylog(n))$ machine space. Thus our algorithm runs in $O(\log(nW)\log\log(n)\cdot(1/\epsilon)^{1/\epsilon})$ rounds and $O(n/polylog(n))$.

\end{proof}

\section{Revenue approximation}

\begingroup
\def\thelemma{\ref{lem:half}}
\begin{lemma}
\half
\end{lemma}
\addtocounter{lemma}{-1}
\endgroup

\begin{proof}
After the first round of merges, note that any duplicated vertex must be merged with its duplicate. This is because the edge weight between these vertices is essentially infinite (for the purposes of this paper, we will say it is arbitrarily large). Thus no duplicates will be matched with another duplicate, so each of the 2-sized clusters has at least one real vertex. In any subsequent merges, this property will hold. Thus it holds for all clusters beyond the initial singleton clusters.
\end{proof}

\begingroup
\def\thelemma{\ref{lem:transform}}
\begin{lemma}
\transform
\end{lemma}
\addtocounter{lemma}{-1}
\endgroup

\begin{proof}
    We start by constructing the vertex set, $V^{i+1}$, which corresponds to the clusters at the new $i+1$th level. So for every $C^{i+1}_j\in C^{i+1}$, create a vertex $v_{C^{i+1}_j}$ and put it in vertex set $V^{i+1}$. It must be a complete graph, so we can add edges between all pairs of vertices.
    
    Consider two vertices $v_{C^{i+1}_j}$ and $v_{C^{i+1}_k}$. Since we merge sets of two clusters at each round, these must have come from two clusters in $C^i$ each. Say they merged clusters from vertices $u_1,u_2$ and $v_1,v_2$ respectively. Note that these vertices are from the previous clustering graph, $\mathcal{C}(G,C^i)$. 
    Then the edges $(u_1,v_1),(u_1,v_2),(u_2,v_1),(u_2,v_2)$  have weights that are the average distances between corresponding $i$-level clusters (because they were from the previous clustering graph, $\mathcal{C}(G,C^i)$). Since the clusters in $C^i$ all have the same size, we can calculate the weight as follows.
    \begin{align*}
        w_{\mathcal{C}(G, C^{i+1})}(v_{C^{i+1}j},v_{C^{i+1}_k}) = \frac14(&w_{\mathcal{C}(G, C^{i})}(u_1,v_1) + w_{\mathcal{C}(G, C^{i})}(u_1,v_2) + w_{\mathcal{C}(G, C^{i})}(u_2,v_1) + w_{\mathcal{C}(G, C^{i})}(u_2,v_2)).
    \end{align*}
    
    This is true because the weights in $\mathcal{C}(G,C^i)$ are already the average weights in the $i$th level clusters, so they are normalized for the clusters size, which is $1/4$th of the cluster size at the next level. So when we sum together the four edge weights, we account for all the edges that contribute to the edge weight in the next level, then we only need to divide by 4 to find the average.
    
    Matching Affinity Clustering can utilize one machine per $(i+1)$-level cluster . Each machine needs to keep track of the distance between its subclusters and all other subclusters at level $i$. It can then do this calculation to capture edge weights in one round with $O(n)$ space.
\end{proof}

\begingroup
\def\thelemma{\ref{lem:rev}}
\begin{lemma}
\lemrev
\end{lemma}
\addtocounter{lemma}{-1}
\endgroup

\begin{proof}
    First, we want to break down the clustering revenue into the sum of its merge revenues. Since each match in our matching $M_i$ defines a cluster in the next level of the hierarchy, we can view each merge as a match in $M_i$. Then we apply the definition of merge revenue.
    \begin{align*}
    \clusterrev_G(C^i, C^{i+1}) &= \sum_{(A,B)\in M_i}\mergerev_G(A,B),
    \\&= \sum_{(A,B)\in M_i}(2^n - |A| - |B|) \sum_{a\in A,b\in B} w_{G}(a,b).
    \end{align*}
Because we start with a power of two vertices after padding, each step can find a perfect matching, thus yielding a power of two many clusters of equal size at each level. Then since cluster size doubles each round, the cluster size at the $i$th iteration is $2^i$. Even though some of the vertices may not contribute to the revenue, this is an upper bound on the size. So $n-|A|-|B|$ in this formula is at least $2^n-2^{i+1}$. This is the work done in (\ref{11-1}) below. 
\begin{align*}
\clusterrev_G(C^i, C^{i+1}) &\geq (2^n - 2^{i+1})\sum_{(A,B)\in M_i}\sum_{a\in A,b\in B} w_G(a,b),\tag{1}\label{11-1}
\\&= (2^n - 2^{i+1})\sum_{(A,B)\in M_i}|A||B|w_{\mathcal{C}(G,C^i)}(v_A,v_B),\tag{2}\label{11-2}
\\&= 2^{2i-2p}(2^n - 2^{i+1})\sum_{(A,B)\in M_i}w_{\mathcal{C}(G,C^i)}(v_A,v_B),\tag{3}\label{11-3}
\\&= 2^{3i - 2p + 1}(2^{n-i-1} - 1)\sum_{(A,B)\in M_i}w_{\mathcal{C}(G,C^i)}(v_A,v_B).\tag{4}\label{11-4}
\end{align*}
In (\ref{11-2}), we simply substitute in place of the sum of the edge weights between $A$ and $B$ in $G$. By definition, the edge weight between $v_A$ and $v_B$ in the clustering graph is the average of the same edge edge weights in $G$. Thus if we just scale that by $|A||B|$, we can substitute it in for the sum of those edge weights. In (\ref{11-3}), we simply pull out $|A||B|$. These contain $2^i$ total vertices, and by Lemma \ref{lem:half}, they contain at least $2^{i-1}$ that contribute to the revenue for a total factor of $2^{2i-2}$. If there were $2^n$ vertices to start, then all clusters contain only real vertices, so the factor is $2^{2i}$. With the indicator, this is $2^{2i-2p}$. We then simplify in (\ref{11-4}).
\end{proof}

\begingroup
\def\thelemma{\ref{lem:costb}}
\begin{lemma}
\costb
\end{lemma}
\addtocounter{lemma}{-1}
\endgroup

\begin{proof}[Proof of Lemma \ref{lem:costb}]
As in Lemma \ref{lem:rev}, we want to break apart the clustering cost at the $i$th level into a series of merge costs. Again, we know the merge costs can be defined through the matching $M_i$.
\begin{align*}
\clustercost_{G}(C^i, C^{i+1}) &= \sum_{(A,B)\in M_i}\mergecost_{G}(A,B),\tag{1}\label{12-1}
\\&= \sum_{(A,B)\in M_i}\left(|A|\sum_{b\in B, c\notin A\cup B} w_{G}(b,c)+ |B|\sum_{a\in A, c\notin A\cup B} w_{G}(a,c)\right),\tag{2}\label{12-2}
\\&\leq 2^i\sum_{(A,B)\in M_i}\left(\sum_{b\in B, c\notin A\cup B} w_{G}(b,c) + \sum_{a\in A, c\notin A\cup B} w_{G}(a,c)\right).\tag{3}\label{12-3}
\end{align*}
At this step, we broke the clustering cost into merge costs of the matching (\ref{12-1}), applied the definition of merge cost (\ref{12-2}), and pulled out $|A|=|B|\leq2^i$ (\ref{12-3}). Consider the inner clusters. Instead of selecting $c\notin A\cup B$, we can consider $c$ being in any other cluster from $C^i$. Let $C_1^i,\ldots,C_k^i\in C^i$ be all clusters other than $A$ or $B$ (i.e., $C^i_j\neq A,B$). Then we can define $c$ as an element in any cluster $C^i_j$ for $j\in[k]$. After, we simply rearrange the indices of summation.
\begin{align*}
\clustercost_{G}(C^i, C^{i+1}) &\leq 2^i\sum_{(A,B)\in M_i}\left(\sum_{b\in B,j\in[k]}\sum_{c\in C_j^i} w_{G}(b,c) + \sum_{a\in A,j\in[k]}\sum_{c\in C_j^i} w_{G}(a,c)\right),
\\&= 2^i\sum_{(A,B)\in M_i}\left(\sum_{j\in[k]}\sum_{b\in B, c\in C_j^i} w_{G}(b,c) + \sum_{j\in[k]}\sum_{a\in A,c\in C_j^i} w_{G}(a,c)\right).
\end{align*}
Recall that the edge weights in $\mathcal{C}(G,C^i)$ are the average edge weights between clusters in $C^i$ on graph $G$. Again, to turn this into just the summation of the edge weights, we must scale by $|B||C^i_j|$ and $|A||C^i_j|$.
\begin{align*}
\clustercost_{G}(C^i, C^{i+1})
&\leq 2^i\sum_{(A,B)\in M_i}\left(\sum_{j\in[k]}|B||C^i_j|w_{\mathcal{C}(G,C^i)}(v_B,v_{C_j})+ \sum_{j\in[k]}|A||C^i_j|w_{\mathcal{C}(G,C^i)}(v_A,v_{C^i_j})\right),
\\&= 2^{3i}\sum_{(A,B)\in M_i}\left(\sum_{j\in[k]}w_{\mathcal{C}(G,C^i)}(v_B,v_{C_j})+ \sum_{j\in[k]}w_{\mathcal{C}(G,C^i)}(v_A,v_{C_j})\right).
\end{align*}

For each iteration of the outer summation, we are taking all the edges with one endpoint as $A$ and all edges with one as $B$ (besides the edge from $A$ to $B$ itself) and adding their weights. Since the only edge in $M_i$ with an endpoint at $A$ or $B$ is the edge from $A$ to $B$, the summation covers all edges with one endpoint as either $A$ or $B$ that are not in $M_i$. Consider an edge from some $C$ to $C'$ that isn't in $M_i$. In every iteration besides possibly the first, $M_i$ matches everything, so we will consider $C$ and $C'$ in separate iterations of the sum. In both of these iterations, we add the weight $w_{\mathcal{C}(G,C^i)}(v_C,v_{C'})$. Thus, each edge in $\mathcal{C}(G,C^i)$ outside of $M_i$ is accounted for twice, and no edge in $M_i$ is accounted for.

Consider a multigraph $H$ with vertex set $V(\mathcal{C}(G,C^i))$ and an edge set that contains all edges \textit{except} those in $M_i$ twice over. Note since $\mathcal{C}(G,C^i)$ was a complete graph, $H$ must be a $2(|V(\mathcal{C}(G,C^i)|-2)$-regular graph. Thus, we could find a perfect matching in $H$ with a maximal matching algorithm, remove those edges to decrease all degrees by 1, and repeat on the new regular graph. Do this until all vertices have degree 0. Since each degree gets decremented by 1 each iteration, there must be a total of $2(|V(\mathcal{C}(G,C^i)|-2)$ matchings $N_1,N_2,\ldots,N_{2(|V(\mathcal{C}(G,C^i)|-2)}$. Thus the clustering cost can be alternatively thought of as the sum of the weights of these alternate matchings in clustering graph $\mathcal{C}(G,C^i)$.

\begin{align*}
&\clustercost_{G}(C^i, C^{i+1}) \\&\qquad\leq 2^{3i}\sum_{j\in[2(|V(\mathcal{C}(G,C^i)|-2)]} w_{\mathcal{C}(G,C^i)}(N_j),\tag{4}\label{12-4}
\\&\qquad\leq 2^{3i}\sum_{j\in[2(|V(\mathcal{C}(G,C^i)|-2)]} w_{\mathcal{C}(G,C^i)}(M_i^*),\tag{5}\label{12-5}
\\&\qquad\leq 2^{3i+1}(|V(\mathcal{C}(G,C^i)|-2) w_{\mathcal{C}(G,C^i)}(M_i^*).\tag{6}\label{12-6}
\end{align*}
In (\ref{12-4}), we viewed the summations as the sum of weights of the alternative matchings described earlier. Step (\ref{12-5}) utilizes the fact that $M_i^*$ is a maximum matching, so the weight of any $N_j$ is bounded above by the weight of $M_i^*$. Finally, in (\ref{12-6}), we note that the summation does not depend on $j$, and so we remove the summation.

Since $M_i$ is an approximation of the maximum matching on $\mathcal{C}(G,C^i)$, we know  $w_{\mathcal{C}(G,C^i)}(M_i^*) \leq w_{\mathcal{C}(G,C^i)}(M_i) / (1-\epsilon)$. We can substitute this in and then rewrite it as the summation of edge weights in $M_i$.

\begin{align*}
\clustercost_{G}(C^i, C^{i+1})&\leq 2^{3i+1}(|V(\mathcal{C}(G,C^i)|-2) \frac1{1-\epsilon} w_{\mathcal{C}(G,C^i)}(M_i),
\\&\leq 2^{3i+1}(|V(\mathcal{C}(G,C^i)|-2)\frac1{1-\epsilon} \sum_{(A,B)\in M_i} w_{\mathcal{C}(G,C^i)}(v_A,v_B).
\end{align*}

The total number of vertices in $\mathcal{C}(G,C^i)$ (ie, the total number of clusters at the $i$th level) is just the total number of vertices over the cluster sizes: $2^n/2^i$. Plugging that in gives the desired result.
\begin{align*}
\clustercost_{G}(C^i, C^{i+1}) &\leq 2^{3i+1}\left(\frac {2^n}{2^i}-2\right) \frac1{1-\epsilon}\sum_{(A,B)\in M_i}w_{\mathcal{C}(G,C^i)}(v_A,v_B),\tag{7}\label{12-i}
\\&\leq 2^{3i+2}\left(2^{n-i-1}-1\right) \frac1{1-\epsilon}\sum_{(A,B)\in M_i}w_{\mathcal{C}(G,C^i)}(v_A,v_B),
\\&\leq \frac{2^{2p+1}}{1-\epsilon} \clusterrev_G(C^i, C^{i+1}).\tag{8}\label{12-8}
\end{align*}
The first steps consist of plugging in the cluster sizes and performing algebraic simplifications. Finally, Step (\ref{12-8}) refers to Lemma \ref{lem:rev} for the clustering revenue. Recall this is an upper bound for the clustering cost in $G$, and so the proof is complete.

So far, we have covered most of the lemma's claim. Now we just need to account for the first step, when we utilize the $k$-sized matching. Note the argument is dependent on $M_i$ being a perfect matching, where $M_0$ may only be a maximum matching on $2N-2^n$ vertices. The proof structure here will function similarly. In this case, we still construct a multigraph $H$ as described on the subset of $G$ containing vertices matching in $M_0$, then we add double copies of all the edges between vertices in the matching that aren't matched to each other for a max degree of $4N-2^{n+1}-4$ and thus create $4N-2^{n+1}-4$ matchings on the $2N-2^n$ vertices to cover these edges. However, the cost also accounts for edges from the matched vertices to the unmatched vertices. We can construct a bipartite graph with all these edges once. Then all vertices on one side of the bipartition have degree $2^n-N$, and the vertices on the other side have degree $2N-2^n$. So we can construct $2^n-N$ matchings with $2N-2^n$ edges in this graph that cover all edges. Alternatively, this can be viewed as $2^{n+1}-2N$ matchings on $2N-2^n$ vertices. In this case, we have a bunch of sized matchings on $2N-2^n$ vertices, $N_1,N_2,\ldots,N_{2N-4}$. The rest of the arguments hold. Since $i=0$, step (\ref{12-i}) becomes the following.

\begin{align*}
\clustercost_G(C^0,C^1) &\leq 2\left(N-2\right)\cdot \frac1{1-\epsilon}\sum_{(A,B)\in M_i}w_{\mathcal{C}(G,C^i)}(v_A,v_B),\tag{9}\label{12-10}
\\&\leq 2(2^{n} - 2)\frac1{1-\epsilon}\sum_{(A,B)\in M_i}w_{\mathcal{C}(G,C^i)}(v_A,v_B),\tag{10}\label{12-11}
\\&\leq 4(2^{n-1} - 1)\frac1{1-\epsilon}\sum_{(A,B)\in M_i}w_{\mathcal{C}(G,C^i)}(v_A,v_B),
\\&\leq \frac{2^3}{1-\epsilon}\clusterrev_G(C^0,C^1).\tag{12}\label{12-13-1}
\end{align*}

This mirrors the computation in steps (\ref{12-i}) through (\ref{12-8}). In (9), we substitute $i$ in for (\ref{12-i}), and also replace the number of matchings with the new number of matchings (though recall at this point, we have already halved that value).  Step (10) applies the fact that $N<2^n$. Finally, in (11), we substitute in the clustering revenue. In the analysis for revenue, note that there do not exist unreal vertices yet, so we can consider $p=0$ when we refer to the Lemma \ref{lem:rev}. However, for this Lemma proof, this is the case where we do eventually duplicate vertices, so we analyze it along with other clusterings where $p=1$, so it only needs to meet the condition when $p=1$.

\begin{align*}
\clustercost_G(C^0,C^1) &\leq \frac{2^{2p+1}}{1-\epsilon}\clusterrev_G(C^0,C^1).\tag{13}\label{12-13}
\end{align*}

Thus concludes our proof.

\end{proof}

\begingroup
\def\thelemma{\ref{lem:apx}}
\begin{lemma}
\lemapx
\end{lemma}
\addtocounter{lemma}{-1}
\endgroup

\begin{proof}
We prove this by constructing Matching Affinity Clustering. Our algorithm starts by allocating one machine to each cluster. Run Algorithm \ref{alg:kmatch} for either the desired $k$- or $n/2$-matching, which finds our $1-\epsilon$ approximate matching, to create clusters of two vertices each, then apply the algorithm from Lemma \ref{lem:transform} to construct the next clustering graph based off this clustering. Repeat this process until we have a single cluster.

From Lemma \ref{lem:costb}, we see that at each round, the cumulative cost is bounded above by $\frac{2}{1-\epsilon}$ times the revenue.
Then  we utilize the definition of the cost of an entire hierarchy tree $T$ to get bounds.
\begin{align*}
\cost_G(T) &= 2\sum_{u,v\in G, u\neq v} w_G(u,v) + \sum_{\text{merges of } A,B} \mergecost_G(A,B),\tag{1}\label{6-1}
\\&= 2\sum_{u,v\in G, u\neq v} w_G(u,v) + \sum_{i\in[\log n]} \clustercost_G(C^i, C^{i+1}),\tag{2}\label{6-2}
\\&\leq 2\sum_{u,v\in G, u\neq v} w_G(u,v) + \frac{2^{3p+1}}{1-\epsilon}\sum_{i\in[\log n]} \clusterrev_G(C^i, C^{i+1}),\tag{3}\label{6-3}
\\&\leq 2\sum_{u,v\in G, u\neq v} w_G(u,v) + \frac{2^{2p+1}}{1-\epsilon}\rev_G(T).\tag{4}\label{6-4}
\end{align*}
Step (1) simply break down the total cost into merge costs, and then step (2) consolidates merge costs in each level of the hierarchy into clustering costs. Note that every iteration halves the number of clusters, so there must be $\log n$ iterations. In (3), we apply the result from Lemma \ref{lem:costb}, and finally in (4), we add up all the clustering revenues into the total hierarchy revenue. We can then examine hierarchy revenue.
\begin{align*}
\rev_G(T) &= n\sum_{u,v\in G, u\neq v} w_G(u,v) - \cost_G(T),
\\&\geq n\sum_{u,v\in G, u\neq v} w_G(u,v) - 2\sum_{u,v\in G, u\neq v} w_G(u,v) - \frac{2^{2p+1}}{1-\epsilon}\rev_G(T).
\end{align*}
The above simply utilizes the duality of revenue and cost, and then substitution from step (4). Next we only require algebraic manipulation to isolate $\rev_G(T)$.
\begin{align*}
&\frac{2^{2p+1} + 1 - \epsilon}{1-\epsilon}\rev_G(T) \geq (n-2)\sum_{u,v\in G, u\neq v} w_G(u,v).
\\&\rev_G(T) \geq \frac{1-\epsilon}{2^{2p+1} + 1 - \epsilon}(n-2)\sum_{u,v\in G, u\neq v} w_G(u,v).
\end{align*}
Then we know the optimal solution $T^*$ can't have more than $n-2$ non-leaves, which means that each edge will only contribute at most $n-2$ times its weight to the revenue. Thus, $\rev_G(T^*) \leq (n-2)\sum_{u,v\in G, u\neq v} w_G(u,v)$. In addition, since $\frac{1-\epsilon}{2^{2p+1}+1-\epsilon}$ can be arbitrarily close to $\frac1{2^{2p+1}+1}$, we rewrite it as $\frac1{2^{2p+1}+1} - \epsilon$. Apply all these to our most recent inequality to get the desired results.
\begin{align*}
\rev_G(T) \geq& \left(\frac1{2^{2p+1}+1}-\epsilon\right)\rev_G(T^*).
\end{align*}

For an input of size $2^n$, we have $p=0$ and get a $\frac13 - \epsilon$ approximation for revenue. For all other inputs, $p=1$, and we get a $\frac19 - \epsilon$ approximation. We note that these applications of cost and revenue properties are heavily inspired by Moseley and Wang's proof for the approximation of Average Linkage \cite{dual}.
\end{proof}

\begingroup
\def\thelemma{\ref{lem:bounds}}
\begin{lemma}
\bounds
\end{lemma}
\addtocounter{lemma}{-1}
\endgroup
\begin{proof}
 First, we use Algorithm \ref{alg:kmatch} to obtain a $k$-sized matching, which runs in $O(\log(nW)\log\log(n)\cdot(1/\epsilon)^{1/\epsilon})$ rounds and $O(n/polylog(n))$ machine space. After this, there are $\log n$ iterations, and at each iteration, we use \citet{matching}'s matching algorithm Algorithm, which finds our $(1-\epsilon)$-approximate matching in $O(\log\log n \cdot (1/\epsilon)^{O(1)})$ rounds and $O(n/polylog(n))$ machine space \cite{matching}. We also transform the graph as in Lemma \ref{lem:transform}, which adds no round complexity, but requires $O(n)$ space. Thus in total, this requires $O(\log(nW)\log (n)\log\log (n) \cdot (1/\epsilon)^{O(1/\epsilon)})$ rounds and $O(n)$ space per machine. However, note that when $p=0$, we can just run \citet{matching}'s algorithm directly, and achieve a slightly better bound of $O(\log (n)\log\log (n) \cdot (1/\epsilon)^{O(1/\epsilon)})$ rounds and $O(n)$ space per machine.
 \end{proof}

\begingroup
\def\thetheorem{\ref{thm:lb}}
\begin{theorem}
\lb
\end{theorem}
\addtocounter{theorem}{-1}
\endgroup

\begin{proof}[Proof of Theorem \ref{thm:lb}]
The graph $G$ consists of $N$ vertices. We divide the vertices into $N^{1/3}$ ``columns'' to make large cliques. In every column, make a clique with edge weights of 1. In addition, enumerate all vertices in each column. For some index $i$, we take all $i$th vertices in each column and make a ``row'' (so there are $N^{2/3}$ rows that are essentially orthogonal to the columns). Rows become cliques as well, with edge weights of $1+\epsilon$. All non-edges are assumed to have weight zero.

This is the graph described by \citet{linkage} to show Average Linkage only achieves a $1/3$-approximation, at best, for revenue. They are able to achieve this because Average Linkage will greedily select all the $1+\epsilon$ edges to merge across first. We can then leverage these results by showing that Matching Affinity Clustering, too, merges across these edges first.

In our formulation, we assume $N=2^{3n}$ for some $n$. Then there are $2^n$ columns and $2^{2n}$ rows with $2^{2n}$ and $2^n$ vertices respectively. Additionally, our algorithm skips the $k$-matching steps (ie, all vertices are real). In the first round, we can clearly find a maximum perfect matching by simply matching within the rows, and we can assure this for our approximate matching by selecting a small enough error. In the next clustering graph, since edge weights are the average linkage between nodes, note that the highest edge weights are still going to be $1+\epsilon$ within the rows. Therefore, this matching will continue occurring until it can no longer find perfect matchings within the rows. Since the rows are cliques of $2^n$ vertices, this will happen until all each row is merged into a single cluster. This is then sufficient to refer to the results of \citet{linkage} to get a $1/3$ bound.
\end{proof}

\section{Value approximation}
Our goal in this section is to prove Theorem~\ref{thm:val}.

\begingroup
\def\thetheorem{\ref{thm:val}}
\begin{theorem}
\matchclusterval
\end{theorem}
\addtocounter{theorem}{-1}
\endgroup

Since this is effectively the same algorithm as the revenue context, we can the analysis of Lemma~\ref{lem:bounds} and simplify it to show the complexity of the algorithm. All that is left to do is prove the approximation. Our proof will be quite similar to that of~\citet{cohen-addad}, however we will require some clever manipulation to handle many merges at a time. Fortunately, the fact that we merge based off a minimum matching with respect to Average Linkage makes the analysis still follow the~\citet{cohen-addad} proof quite well. We start with a lemma.

\begin{lemma}
\vallem
\end{lemma}

\begin{proof}
Let $a = \frac12\sum_{i=1}^k|A_i|(|A_i|-1)$ and $b = \frac12\sum_{i=1}^k|B_i|(|B_i|-1)$. Let $A = \cup_{i=1}^k A_i$ and $B = \cup_{i=1}^k B_i$. Using these, one can see that the average edge weight of all edges contained in any $A_i$ or $B_i$ cluster is: \[\frac{\sum_{i=1}^k (w(A_i) + w(B_i))}{a + b}\]
These edges were all merged across at some point lower in the hierarchy. This means that the edge set between $A_i$'s and $B_i$'s are the union of all edges merged across in lower clusterings in the hierarchy. Therefore, by averaging, we can say there exists a clustering $C'$ (with $|C'|=k'$) below $C$ in the hierarchy, with clusters and splits similarly defined as $C'$, $A_i'$ and $B_i'$ respectively, such that:
\[\frac{\sum_{i=1}^{k'} w(A_i',B_i')}{\sum_{i=1}^k|A_i'|\cdot |B_i'|} \geq \frac{\sum_{i=1}^k (w(A_i) + w(B_i))}{a + b}\]
Now we would like to build a similar expression for the edges between all $A_i$ and $B_i$. The average of these edge weights is the following expression:
\[\frac{\sum_{i=1}^kw(A_i, B_i)}{\sum_{i=1}^k |A_i| \cdot |B_i|}\]
Consider the iteration that formed $C'$. Notice because $C$ is a higher level of the hierarchy, every cluster $A_j'$ and $B_j'$ must be a subset of some $A_i$ or $B_i$. Fix some $i$, and consider all the edges that cross from some $A_j'$ to some $B_k'$ such that $A_j', B_k' \subset A_i\cup B_i = C_i$. The union of all these edges precisely makes up the set of edges between $A_i$ and $B_i$. Do this for every $i$, and we can see this makes up all the edges of interest. We can decompose this into a set of matchings across the entire dataset. By another averaging argument, we can say there exists another alternate clustering $C''$ (as opposed to $C'$) which only matches clusters $A_j'$ and $B_k'$ that are descendants of $A_i$ and $B_i$ respectively such that:
\[\frac{\sum_{i=1}^kw(A_i'', B_i'')}{\sum_{i=1}^k |A_i''| \cdot |B_i''|} \leq \frac{\sum_{i=1}^kw(A_i, B_i)}{\sum_{i=1}^k |A_i| \cdot |B_i|}\]
Note that this was a valid matching, and therefore a valid clustering, at the same time that $C'$ was selected. Also, note that either both $C$ and $C'$ were perfect matchings, or they were both restricted to the same $k$ size in the first step of the algorithm. Thus, since $C'$ was an $\alpha$-approximate minimum ($k$-sized) matching in the graph where edges are the average edge weights between clusters, we know:
\[\frac{\sum_{i=1}^kw(A_i', B_i')}{\sum_{i=1}^k |A_i'| \cdot |B_i'|} \leq \alpha\frac{\sum_{i=1}^kw(A_i'', B_i'')}{\sum_{i=1}^k |A_i''| \cdot |B_i''|}\]
Putting this all together, we find the desired result:
\[\frac{\sum_{i=1}^k (w(A_i) + w(B_i))}{a + b} \leq \alpha\frac{\sum_{i=1}^kw(A_i, B_i)}{\sum_{i=1}^k |A_i| \cdot |B_i|}\]
\end{proof}

And we can use this to prove our theorem, similar to the results of~\citet{cohen-addad}.

\begin{proof}[Proof of Theorem~\ref{thm:val}]
We prove this by induction on the level of the tree. At some level, with clustering $C$, consider truncating the entire tree $T$ at that level, and thus only consider the subtrees \textit{below} $C$, ie with roots in $C$. Call this tree $T_C$. We will consider the aggregate value accumulated by this level.  Trivially, the base case holds. Then we can split the value of $C$ into the value accumulated at the most recent clustering step and value one step below $C$. We use induction on the latter value. Since the approximation ratio is $\frac13$ for $p=1$ and $\frac12$ for $p=0$, we can write the ratio as $\left(\frac13\right)^p \left(\frac12\right)^{1-p}$.
\begin{align*}
\val(T_C) \geq&  \sum_{i=1}^k (|A_i| + |B_i|) w(A_i,B_i) + \left(\frac13\right)^p \left(\frac23\right)^{1-p}\sum_{i=1}^k(|A_i|w(A_i) + |B_i|w(B_i))
\end{align*}
Now we would like to apply Lemma 1 to modify the first term in a similar manner to Cohen-Addad et al. Specifically, we want to extract terms of the form $|A_i|w(B_i)$ and $|B_i|w(A_i)$. We will find this is harder to do with our formulation of Lemma 1, and therefore we have to rely on Lemma x that says that the cluster balance is at least $\frac12$. Let $m=\min\left(\{|A_i|\}_{i=1}^k \cup\{|B_i|\}_{i=1}^k\}\right)$ be the minimum cluster size. This and our cluster balance ratio implies that $m\leq  |A_i|,|B_i| \leq 2m$ for all $i$. Note that when $N=2^n$, we have perfect cluster balance, so $|A_i|=|B_i| = m$. Thus for our indicator $p$, $m\leq |A_i|,|B_i| \leq 2^pm$. Now we can manipulate our Lemma 1 result. Start by isolating the numerator on the left.
\begin{align*}
&\sum_{i=1}^k w(A_i, B_i) \geq 2\alpha\frac{\sum_{i=1}^k |A_i|\cdot |B_i|\sum_{i=1}^k (w(A_i) + w(B_i))}{ \sum_{i=1}^k (|A_i|(|A_i| -1 ) + |B_i|(|B_i|-1))}
\end{align*}
Note now that $\sum_
{i=1}^k |A_i|\cdot |B_i| \geq m^2$ and $|A_i|(|A_i|-1) \leq 2^{2p}m^2$  and similarly for $B_i$.
\begin{align*}
\sum_{i=1}^k w(A_i, B_i) \geq& 2\alpha\frac{km^2\sum_{i=1}^k (w(A_i) + w(B_i))}{2^{2p+1}km^2}
\\=& \alpha\frac{\sum_{i=1}^k (w(A_i) + w(B_i))}{2^{2p}}
\end{align*}
To get the correct term on the left, we see that $\sum_{i=1}^k (|A_i| + |B_i|)w(A_i,B_i) \geq 2m\sum_{i=1}^k w(A_i,B_i)$. So we can multiply this result by $2m$, and then plug it into a portion of the first term. To preserve the ratio for both $p=1$ and $p=0$, we multiply it by $\left(\frac23\right)^p \left(\frac12\right)^{1-p}$.

\begin{align*}
\val(T_C) \geq& \left(1-\frac23\right)^{1-p} \left(1-\frac13\right)^{p}\sum_{i=1}^k (|A_i| + |B_i|) w(A_i,B_i) \\&+ \frac{2m}{2^{2p}}\cdot\left(\frac23\right)^p \left(\frac13\right)^{1-p}\alpha\sum_{i=1}^k(w(A_i) + w(B_i)) \\&+ \left(\frac13\right)^p \left(\frac23\right)^{1-p}\alpha\sum_{i=1}^k(|A_i|w(A_i) + |B_i|w(B_i))
\end{align*}
Next, note $m\geq |A_i|, |B_i|$. This can be used on the second term.
\begin{align*}
\val(T_C) \geq& \left(\frac13\right)^{p} \left(\frac23\right)^{1-p}\sum_{i=1}^k (|A_i| + |B_i|) w(A_i,B_i) \\&+ \cdot\left(\frac13\right)^p \left(\frac23\right)^{1-p}\alpha\sum_{i=1}^k(w(A_i) + w(B_i)) \\&+ \left(\frac13\right)^p \left(\frac12\right)^{1-p}\alpha\sum_{i=1}^k(|A_i|w(A_i) + |B_i|w(B_i))
\\\geq& \left(\frac13\right)^p\left(\frac23\right)^{1-p} \alpha\sum_{i=1}^k |C_i| w(C_i)
\end{align*}
Therefore, this captures $\frac23\alpha$ of the weight of each subtree at height $C$ when $n=2^N$, and $\frac13\alpha$ more generally.
\end{proof}

Finally, we can show Theorem~\ref{thm:lbval}, which shows the tightness of the stronger approximation factor.

\begingroup
\def\thetheorem{\ref{thm:lbval}}
\begin{theorem}
\lbval
\end{theorem}
\addtocounter{theorem}{-1}
\endgroup

\begin{proof}
Consider $G$ which is almost a bipartite graph between partitions $A$ and $B$ (with $|A|=|B|$), except with a single perfect matching removed. For instance, if we enumerate $A=\{a_1,\ldots,a_n\}$ and $B=\{b_1,\ldots,b_n\}$, we have $w(a_i,b_j)=1$ for all $i\neq j$ and $w(a_i,b_i) =0$. And since it's bipartite, $w(a_i,a_j)=w(b_i,b_j)=0$ for all $i, j$.
\\\\Consider the removed perfect matching to get clusters $\{a_1,b_1\},\ldots,\{a_n,b_n\}$. Matching Affinity Clustering could start by executing these merges, as this is a zero weight (and thus minimum) matching.
\\\\Now consider $G'$, the remaining graph after these merges with a vertex for each cluster and edges representing the total edge weight between clusters. This is a complete graph of size $n$ with 2-weight edges. By \citet{dasgupta}'s results, we know the value (with is calculated the same as cost) of any hierarchy on $G'$ is $\frac23(n^3-n)$. However, this ignores the fact that clusters are size 2, so the contribution of this part to the hierarchy on $G$ yields a revenue of $\frac43(n^3-n)$.
\\\\Now let's consider the obvious good hierarchy, where we simply merge all of $A$, then all of $B$, then merge $A$ and $B$ together. Thus all $n(n-2)$ edges will be merged into a cluster of size $2n$ for a total value of $2(n^3-2n^2)$.
\\\\Asymptotically, then, Matching Affinity Clustering only achieves a ratio of $2/3$ on this graph.
\end{proof}

\section{Affinity Clustering approximation bounds}
 
We now formally prove our bounds on the theoretical performance of Affinity Clustering.

\begingroup
\def\thetheorem{\ref{thm:affinity}}
\begin{theorem}
\affinity
\end{theorem}
\addtocounter{theorem}{-1}
\endgroup

Recall that this is shown with the following two lemmas:

\begingroup
\def\thelemma{\ref{lem:affinityrev}}
\begin{lemma}
\affinityrev
\end{lemma}
\addtocounter{lemma}{-1}
\endgroup
 
\begin{proof}
Consider a complete bipartite graph $G$ with $2^n$ vertices such that each partition, $L$ and $R$, has $2^{n-1}$ vertices, and all edges have weight 1. To make it a complete graph, we simply fill in the rest of the graph with 0 weight edges. We first consider how Affinity Clustering might act on $G$. To start, each vertex reaches across one if its highest weight adjacent edges, a weight one edge, and merges with that vertex. Therefore, a vertex in $L$ will merge with a vertex in $R$, and vice versa. There are many ways this could occur. We consider one specific possibility.

Take vertices $c_L$ and $c_R$ from $L$ and $R$ respectively. It is possible that every vertex in $L\setminus c_L$ will merge with $c_R$ and every vertex in $R\setminus c_R$ will merge with $c_L$. It doesn't matter which vertices $c_R$ and $c_L$ try to merge with. Then $G$ is divided into two subgraphs, both of which are stars centered at $c_L$ and $c_R$ respectively. The spokes of the stars have unit edge weights, and all other edges have weight 0. In the next step, the two subgraphs will merge into one cluster. Since there are no non-leaves at that point, it contributes nothing to the total revenue. Therefore, all revenues are encoded in the first step.

Since the subgraphs are identical, they will contribute the same amount to the hierarchy revenue. Recall that we are trying to prove a bound for whatever method Affinity Clustering might choose to break down a merging of a large subgraph into a series of independent clusters. Notice, however, due to the symmetries of the subgraphs, it does not matter in what order the independent merges occur. Therefore, we consider an arbitrary order. Let $T$ be the hierarchy of Affinity Clustering with this arbitrary order. We must break it down into individual merges, and let $T_0$ be the portion of the hierarchy contributing to one of the stars. We only need to sum over the merges of nonzero weight edges.

\begin{align*}
\rev_G(T) =& 2\rev_G(T_0).
\end{align*}

At each step of merging the star subgraph, we merge a single vertex across a unit weight edge into the cluster containing the star center. Call this growing cluster $C_i$ at the $i$th merge. Let $v_i$ be the vertex that gets merged with $C$ at the $i$th step. Then we can break this down into $2^{n-1}-1$ total merges.

\begin{align*}
\rev_G(T) =& 2\sum_{i=1}^{2^{n-1}-1} merge\text{-}\rev_G(\{v_i\}, C_i),\tag{1}\label{13-1}
\\=& 2\sum_{i=0}^{2^{n-1}-2} 2^n - i  - 1,\tag{2}\label{13-2}
\\=& O(2^{2n}).\tag{3}\label{13-3}
\end{align*}

In step (1), we simply break a single star's merging into a series of individual merges in temporal order. Step (2) uses the fact that there are $2^n$ total vertices and $i$ and 1 vertices in the groups being merged to apply the definition of merge revenue. Finally, in (3), we simply evaluate the summation.

Now we consider how Matching Affinity Clustering will act on this graph. It simply finds the maximum matching. In the first iteration, it must match across unit weight edges, and therefore is a perfect matching on the bipartite graph. After this, due to symmetry, it simply finds any perfect matching at each iteration until all clusters are merged. Since the number of vertices is $2^n$, it can always find such a perfect matching. Let $T'$ be Matching Affinity Clustering's hierarchy on $G$. We will break this down into clusterings at each level, as we did in Section \ref{sec:results}. Note there are $\log_2(2^n) = n$ total clusterings required in the hierarchy. And at each clustering, we have a matching $M_i$ that we merge across, so we can break it down into merges across matches. Let any $C^i$ be the usual clustering at the $i$th level of this algorithm.

\begin{align*}
\rev_G(T') =& \sum_{i=1}^{n} \clusterrev_G(C^i, C^{i+1}),
\\=& \sum_{i=1}^{n} \sum_{(A,B)\in M_i} \mergerev_G(A,B).
\end{align*}

Note that by symmetry, each merge on a level contributes the same amount to the revenue. Therefore, we can simply count the number of merges and their contributions. At the $i$th level, there are $2^{n-i}$ total merges. The size of each cluster being merged is $2^{i-1}$, so the number of non-leaves is $2^n-2\cdot 2^{i-1} = 2^n-2^i$. Finally, the number of edges being merged across at each merge, since the graph is bipartite, is just $2^i$.

\begin{align*}
\rev_G(T') =& \sum_{i=1}^{n} 2^{i-1}(2^n - 2^i)\cdot 2^i
= O(2^{3n}).
\end{align*}

Therefore, if we take the ratio of the revenues in the long run, we find the following.

\begin{align*}
\frac{\rev_G(T')}{\rev_G(T)} = O(2^n).
\end{align*}

So by cleverly selecting $n$, we can make the clustering found by Affinity Clustering arbitrarily worse. If the size of the graph is $N=2^n$, then Affinity Clustering can only achieve at best a $1/O(n)$ approximation for this graph.

\end{proof}

\begingroup
\def\thelemma{\ref{lem:affinityval}}
\begin{lemma}
\affinityval
\end{lemma}
\addtocounter{lemma}{-1}
\endgroup

\begin{proof}
Consider simply a graph $G$ that is a matching (ie, each vertex is connected to exactly one other vertex with edge weight 1) with $4n$ vertices. Again, recall that Affinity Clustering must match along the edges of a minimum spanning tree.

Partition the vertices into sets of four, which consist of two pairs. Consider one such set: $v_1$ is matched to $v_2$ and $u_1$ is matched to $u_2$. Most of the edges here are zero. Therefore, a potential component of the minimum spanning tree is the line $v_1,u_1,u_2,v_2$. If we do this for all sets, we can then simply pick an arbitrary root for each set (ie, $v_1$), make some arbitrary order of sets, and connect the roots in a line. All of these added edges are weight 0, so this is clearly a valid MST.

However, note that each edge is contained within some set of four. So if Affinity Clustering merges across these edges first, then the largest cluster the 1-weight edges can be merged in has four vertices. Say the tree returned by Affinity Clustering is $T$. Note that we have $2n$ edges.
\[val(T) \leq 4 \sum_{i=1}^{2n} 1 = 8n.\]
We now observe the optimal solution. Since this is bipartite, we can simply merge each side of the partition first. Then we can merge the two partitions together at the top of the hierarchy. This means all edges are merged into the final, $4n$-sized cluster. Call this $T^*$.
\[val(T^*)  = 4n\sum_{i=1}^{2n} 1 = 8n^2.\]
Thus $val(T)/val(T^*) = 1/O(n)$. Thus, Affinity Clustering cannot achieve better than a $1/O(n)$ approximation on this family of graphs.
\end{proof}

\subsection{Experiments}
Here we provide the full depiction of all experiments.

\begin{figure*}\label{fig:a}
\centering 
\subfigure[Legend]{\includegraphics[width=50mm]{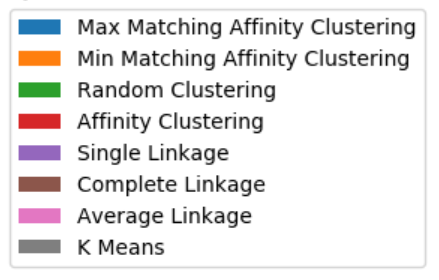}}
\subfigure[Rand Index on Raw Data ]{\label{fig:suba1}\includegraphics[width=110mm]{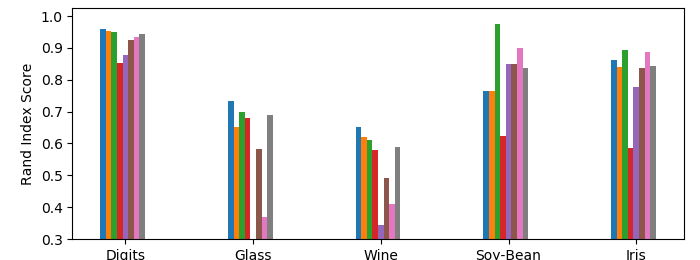}}
\subfigure[Rand Index on Filtered Data]{\label{fig:suba2}\includegraphics[width=110mm]{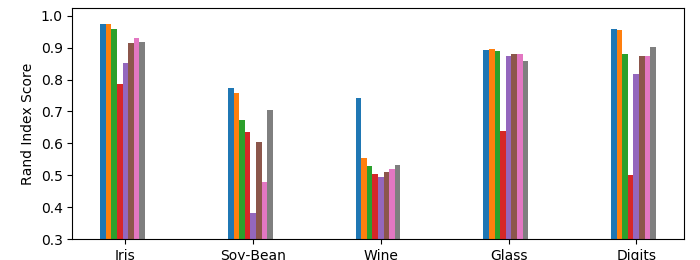}}
\subfigure[Cluster Balance on Raw Data]{\label{fig:suba3}\includegraphics[width=110mm]{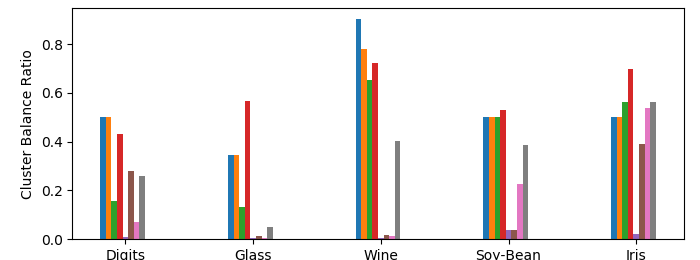}}
\subfigure[Cluster Balance on Filtered Data]{\label{fig:suba4}\includegraphics[width=110mm]{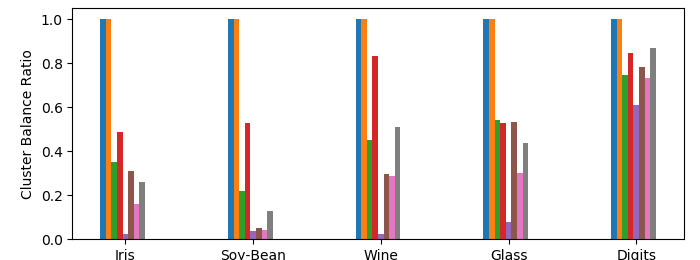}}
\caption{Rand Index scores and cluster balance on raw and filtered (randomly pruned so ground truth is balanced, $n=2^N$) UCI datasets.}
\end{figure*}

\end{document}